\newif\ifabstract
\abstracttrue
 \abstractfalse 
\newif\iffull
\ifabstract \fullfalse \else \fulltrue \fi

\documentclass[11pt]{article}
\usepackage{amsfonts}
\usepackage{amssymb}
\usepackage{amstext}
\usepackage{amsmath}
\usepackage{xspace}
\usepackage{graphicx}
\usepackage{url}
\usepackage{graphics}
\usepackage{colordvi}
\usepackage{colordvi}
\usepackage{subfigure}

\advance \oddsidemargin by -0.8in
\newcommand{\myparskip}{3pt}
\parskip \myparskip

\usepackage[utf8]{inputenc}
\usepackage{mathrsfs}
\usepackage{amsmath,amsthm,bm}
\usepackage{xspace}
\usepackage{latexsym}
\usepackage{graphicx}
\usepackage{epsfig}
\usepackage{subfigure}
\usepackage{verbatim}
\usepackage{cases}
\usepackage{url}
\usepackage{amsfonts,latexsym}
\usepackage{setspace}
\usepackage{subfigure}
\usepackage{tikz}
\usepackage{pgfplots}
\usepackage{footnote}
\usepackage{enumitem}

\usepackage{algorithm}
\usepackage[]{algpseudocode}
\makeatletter
\newcommand{\algmargin}{\the\ALG@thistlm}
\makeatother
\newlength{\whilewidth}
\algdef{SE}[parWHILE]{parWhile}{EndparWhile}[1]
{\parbox[t]{\dimexpr\linewidth-\algmargin}{%
		\hangindent\whilewidth\strut\algorithmicwhile\ #1\ \algorithmicdo\strut}}{\algorithmicend\ \algorithmicwhile}%
\algnewcommand{\parState}[1]{\State%
	\parbox[t]{\dimexpr\linewidth-\algmargin}{\strut #1\strut}}

\makeatletter

\newcommand{\Rmnum}[1]{\expandafter\@slowromancap\romannumeral #1@}
\makeatother

\usepackage{array}
\usepackage{color}

\newtheorem{thm}{Theorem}[section]
\newtheorem{lem}[thm]{Lemma}
\newtheorem{pro}[thm]{Proposition}
\newtheorem{cor}[thm]{Corollary}

\newtheorem {dfn}[thm]{Definition}

\newtheorem{rem}[thm]{Remark}
\newtheorem{ass}[thm]{Assumption}
\newtheorem{claim}[thm]{Claim}

\newcommand{\calm}{\mathcal{M}}     
\newcommand{\cals}{\mathcal{S}}     
\newcommand{\caln}{\mathcal{N}}

\newcommand{\cali}{\mathcal{I}}
\newcommand{\cala}{\mathcal{A}}

\newcommand{\caly}{\mathcal{Y}}

\newcommand{\la}{{\lambda}}

\newcommand{\by}{\boldsymbol{y}}

\newcommand{\bmu}{\boldsymbol{\mu}}
\newcommand{\bla}{\boldsymbol{\lambda}}
\newcommand{\bga}{\boldsymbol{\gamma}}

\newcommand{\got}{\textsf{GOT}\xspace}

\newcommand{\opt}{\textsf{OPT}\xspace}
\newcommand{\alg}{\textsf{ALG}\xspace}

\newcommand{\ota}{\textsf{OTA}\xspace}

\newcommand{\omkp}{\textsf{OMKP}\xspace}
\newcommand{\fomkp}{\textsf{FOMKP}\xspace}
\newcommand{\omkpar}{\textsf{OMKPAR}\xspace}

\newcommand{\opd}{\textsf{OPD}\xspace}
\newcommand{\dual}{\textsf{Dual}\xspace}

\DeclareMathOperator*{\argmax}{arg\,max}

\def\smallint{\begingroup\textstyle \int\endgroup}

\begin{document}

\title{Competitive Algorithms for the Online Multiple Knapsack Problem with Application to Electric Vehicle Charging}

\author{Bo~Sun\thanks{The Hong Kong University of Science and Technology. Email: {\tt bsunaa@connect.ust.hk}.}\and 
Ali~Zeynali\thanks{University of Massachusetts Amherst. Email: {\tt azeynali@umass.edu}.} \and
Tongxin~Li\thanks{California Institute of Technology. Email: {\tt tongxin@caltech.edu}.} \and
Mohammad~Hajiesmaili\thanks{University of Massachusetts Amherst. Email: {\tt hajiesmaili@cs.umass.edu}.} \and
Adam~Wierman\thanks{California Institute of Technology. Email: {\tt adamw@caltech.edu}.} \and
Danny~H.K.~Tsang\thanks{The Hong Kong University of Science and Technology. Email: {\tt eetsang@ust.hk}.}
}

\begin{titlepage}
\maketitle

\thispagestyle{empty}

\begin{abstract}
We introduce and study a general version of the fractional online knapsack problem with multiple knapsacks, heterogeneous constraints on which items can be assigned to which knapsack, and rate-limiting constraints on the assignment of items to knapsacks. This problem generalizes variations of the knapsack problem and of the one-way trading problem that have previously been treated separately, and additionally finds application to the real-time control of electric vehicle (EV) charging. We introduce a new algorithm that achieves a competitive ratio within an additive factor of one of the best achievable competitive ratios for the general problem and matches or improves upon the best-known competitive ratio for special cases in the knapsack and one-way trading literatures. Moreover, our analysis provides a novel approach to online algorithm design based on an instance-dependent primal-dual analysis that connects the identification of worst-case instances to the design of algorithms. Finally, we illustrate the proposed algorithm via trace-based experiments of EV charging.
\end{abstract}

\end{titlepage}

\section{Introduction}
\label{sec:intro}

Online optimization has become a foundational piece of the design of networked and distributed systems that is used to capture the challenges of decision-making in uncertain environments. Theoretical results have had impact for data center optimization~\cite{adnan2012energy,liu2014greening,luo2013temporal,yang2020online}, video streaming~\cite{spiteri2020bola,wang2016competitive}, energy systems~\cite{alinia2018competitive,Alinia2020online,Guo2017optimal,qureshi2009cutting,Zhang2017}, cloud management~\cite{lucier2013efficient,zheng2016online,Zhang2017,tan2020online_combinatorial}, and beyond.  

Two classical problems within the online optimization literature that have received considerable attention in recent years are the online knapsack problem and the one-way trading problem. In the \emph{online knapsack problem}, an agent must make irrevocable decisions about which items to pack into a knapsack without knowing which items will arrive in the future. In the \emph{one-way trading problem}, an investor must trade a limited amount of one asset to another asset without knowing the future conversion rates. These problems have seen broad application in recent years, e.g., to auction-based resource provisioning in cloud/edge clusters~\cite{Zhang2017,tan2020online_combinatorial}, admission control and routing of virtual circuits~\cite{buchbinder2009design,paris2016online}, and transactive control of distributed energy resources~\cite{alinia2019online,Alinia2020online,alinia2018competitive,sun2020orc}.

These two problems are seemingly very different, and the papers on each tend to use very different algorithmic approaches and analytic techniques, e.g., threat-based algorithms~\cite{El2001,fujiwara2011average}, threshold-based algorithms~\cite{Zhang2017,zhou2008budget,tan2020}, online primal-dual algorithms~\cite{buchbinder2009design,sun2020orc,buchbinder2009online}, online linear programming~\cite{agrawal2014dynamic,yang2019online}, model predictive control~\cite{Guo2017optimal,Lee2018large}, and more.
In addition, within each problem, a wide range of variations have been considered, each motivated by features of different applications. 
For example, versions of online $0/1$ knapsack~\cite{Zhang2017}, online multiple knapsacks, where items can be assigned across multiple knapsacks~\cite{zhou2008budget}, and online fractional knapsack, where each item can be partially admitted~\cite{noga2005online}. Similarly, a wide set of variants of one-way trading have emerged, e.g., with~\cite{El2001} or without~\cite{yang2019online} leftover assets, and concave returns~\cite{Lin2019}. The disconnected nature of these literatures begs the question: \emph{Is it possible for a unified algorithmic approach to be developed or does each variant truly require a carefully crafted approach?}  

\textbf{Contributions of This Paper.}
Despite the differences in approaches and techniques, there are also similarities between these problems that lead one to hope that unification is possible. In this paper, we provide such an algorithmic unification via a generalization of the online knapsack and one-way trading problems, i.e., we show that a single algorithmic approach can be used to provide near optimal algorithms across nearly all previously considered variants of these two problems. 

More specifically, we take motivation from the online electric vehicle (EV) charging problem, which is a prominent problem in energy systems~\cite{alinia2018competitive,alinia2019online,Alinia2020online,Guo2017optimal,Lee2018large,sun2020orc,zheng2014online}. 
In this problem, an operator of an EV charging facility must charge a set of EVs that arrive over time without knowing requests of future arrivals.
Each request includes a charging demand, a charging rate limit, and a departure time before which it must receive charge.
Upon the arrival of each EV, the operator receives its request and schedules its charging to maximize the aggregate value of all EVs.
Most commonly, the operator tries to satisfy EVs' charging demands using a \textit{best effort} policy although drivers desire an on-arrival commitment, which notifies them a guaranteed amount of energy to be delivered upon their arrivals~\cite{alinia2019online,sun2020orc}. 
However, on-arrival commitment adds significant challenges to design online algorithms with theoretical guarantee since it introduces strong temporal coupling for the schedule of each EV, especially when a charging rate limit is also enforced. 
Thus, to mitigate the temporal coupling, this paper focuses on a policy that achieves on-arrival commitment by determining a committed schedule upon arrivals.
We show that this simple policy does not lose too much flexibility in the worst case since it can achieve a nearly optimal competitive ratio, and hence provides a good baseline schedule for further adjusting the charging adaptively over time (see Section~\ref{subsec:examples}). 
Because of the rate constraint, the EV charging problem cannot fit into either the one-way trading or the online multiple knapsack problem (\omkp) directly. However, it can be modeled as a form of a fractional \omkp (\fomkp) with rate constraints that generalizes existing problems in both online knapsack and one-way trading.  Additionally, the resulting problem also captures other applications, such as classical formulations of cloud scheduling~\cite{lucier2013efficient,zheng2016online} and geographical load balancing~\cite{qureshi2009cutting,liu2014greening,adnan2012energy,luo2013temporal} (see Section \ref{subsec:examples}).    

Focusing on this new \fomkp, the goal of the paper is to design algorithms that can achieve nearly the optimal value as the optimal algorithm. Specifically, we aim to develop algorithms that maintain a minimal \emph{competitive ratio}, which is the worst-case ratio of the value of the optimal offline algorithm to that achieved by the online algorithm.

To that end, we focus on a form of algorithms introduced by~\cite{zhou2008budget} in the context of \omkp called online threshold-based algorithms (\ota). The design of this class of algorithms is based on a threshold function $\phi$ that estimates the dual variables of the problem based on the knapsack utilization.  See Section \ref{subsec:ota} for a formal introduction to \ota. While \ota has proven effective in some contexts, the application of the approach is limited due to the fact that designing the threshold function $\phi$ is more art than science, similarly to the difficulties in designing Lyapunov functions for Lyapunov-based control~\cite{qu2018exponential,chen2011convergence} and designing potential functions for the analysis of online scheduling algorithms~\cite{abernethy2009competing,grove1991harmonic}.

In this paper, we present a new systematic approach for designing the threshold functions in \ota. The approach, described in Section \ref{subsec:approach}, uses a novel instance-dependent online primal-dual analysis to design the threshold function directly from a characterization of worst-case instances of the problem. Thus, the task of identifying instances in order to prove a lower bound is unified with the task of designing an algorithm that can (nearly) achieve that bound.  

This new approach yields the design of a threshold function for \ota that provides the first algorithm with a competitive ratio within an additive factor of one of the best achievable competitive ratio for the general problem and matches or improves on the best-known competitive bounds for a wide variety of special cases in the knapsack and one-way trading literatures (see Section \ref{subsuc:summary-of-results}). Specifically, we illustrate the approach for classical one-way trading and two of its recent variants. In all cases the approach yields either the optimal competitive ratio or a competitive ratio that improves upon the state-of-the-art.  

Finally, to illustrate the performance of the algorithm in a specific application, we end the paper  with a brief discussion in Section~\ref{sec:exp} of the problem that motivates our study: EV charging with on-arrival commitment to drivers and rate constraints. Note that the on-arrival commitment of a charging level to drivers is a distinctive feature of this case study that adds significant additional challenges compared to typical papers on online EV charging.  Additionally, it is rare for algorithms for online EV charging to have theoretical guarantees when rate constraints are considered. We present a case study using the Adaptive Charging Network Dataset, ACN-Data, which includes 50,000 EV charging sessions \cite{LeeLiLow2019a}. Here, we show that our algorithm, which uses an adaptive utilization-based threshold, improves over the most common prior approaches for related online knapsack problems such as~\cite{El2001}, which use a fixed threshold policy. Our design targets the worst-case performance, and we see an over 40\% decrease (nearly a factor of 2 improvement) in the worst-case when utilization is high, while also achieving an around 20\% decrease on average.

In summary, in this paper we make the following  contributions.
\begin{itemize}
	\item We introduce and study a generalization of the fractional online multiple knapsack problem (\fomkp) that is motivated by the EV charging problem and unifies the online knapsack and one-way trading literatures.   
	\item We develop an approach for designing online threshold-based (\ota) algorithms based on a novel instance-dependent online primal-dual analysis that connects the characterization of worst-case instances to the design of online algorithms.  
	\item We design an algorithm for the general \fomkp problem with rate constraints that has a competitive ratio within an additive factor of one from the optimal competitive ratio. The algorithm also matches or improves upon best-known results in specific cases covered by 
	recent papers, e.g.,~\cite{yang2019online,Lin2019,Zhang2017,zhou2008budget}.   
	\item We illustrate the performance of the algorithm in the context of EV charging using a trace-based case study, showing a decrease in the worst case by up to 44\% and of the average case by around 20\% as compared to fixed threshold policies, the most common approach in prior work on online knapsack problems.  
\end{itemize}  

{\bf Related Work.} The online optimization problem considered in this work is related to, and unifies, problems that have originated from a wide range of applications. In the following, we briefly overview the problems that can be considered as variants or special cases of the \fomkp that we study. Details of these problems are discussed in Section~\ref{subsec:examples}, where the relationship to the \fomkp we study is shown formally.

\textit{The Online Knapsack Problem.} Our work generalizes a class of problems known as the online knapsack problems (OKPs)~\cite{chakrabarty2008online,zhou2008budget,Zhang2017,tan2020}, which are online variants of the well-studied knapsack problem~\cite{pisinger2005hard}. Since there is no competitive online algorithm for general OKPs, studies typically assume that the weight of each item is small and the value-to-weight ratio is bounded from both below and above. Under this \emph{infinitesimal assumption}, competitive algorithms can be derived. For example, $(1+\ln\theta)$-competitive online algorithms have been designed for the classical online $0/1$ knapsacks~\cite{Zhang2017,zhou2008budget}, where $\theta$ is the ratio of the upper and lower bounds of the value-to-weight ratio.

An important generalization of the classical problem is to the case of multiple knapsacks.  In an online multiple knapsack problem (OMKP), an operator has a set of knapsacks with heterogeneous capacities. Items  arrive sequentially, each with an associated weight and value, and an operator decides whether to accept each item and where to pack it if it is accepted. Here, \cite{zhou2008budget} presents an algorithm that can achieve a competitive ratio of $1+\ln\theta$ under the infinitesimal assumption.

A more general version of \omkp is the online multiple knapsack problem with assignment restrictions (\omkpar)~\cite{Kellerer2004}. In this problem, each item is associated with a subset of knapsacks and is restricted to be packed in this subset. Here progress has not been made, even under the infinitesimal assumption.  In this paper, we study a generalization of a \emph{fractional} online multiple knapsack problem (\fomkp), where fractional refers to the fact that items can be assigned such that a fraction goes to each of multiple knapsacks. Fractional assignment is an important feature of many applications, e.g., EV charging~\cite{Alinia2020online,zheng2014online}, cloud scheduling~\cite{lucier2013efficient,zheng2016online}, geographical load balancing~\cite{qureshi2009cutting,liu2014greening}. Furthermore, there is a strong connection between the fractional version of knapsack problems and the integral version with an infinitesimal assumption. In particular, algorithms for fractional versions of the problem also can be used for the integral case under the infinitesimal assumption (see Section~\ref{subsec:ota}). Prior work on fractional knapsack problems includes \cite{noga2005online,iwama2002removable,lueker1998average}. However, none of these works include rate constraints, which are core to EV scheduling.  

There are also other variants of generalized OKPs that have been considered in the literature, such as the online fractional packing problem~\cite{buchbinder2009online} and the online multi-dimensional knapsack problem~\cite{Zhang2017}. In these settings, the best-known competitive ratios depend on the number of knapsacks $M$. For example,~\cite{buchbinder2009online} gives a competitive ratio $O(\ln M)$ for the online fractional packing problem and the state-of-the-art competitive ratio for the online multi-dimensional knapsack problem is $O(M)$, shown in~\cite{Zhang2017}. This variant of the OKPs is harder than the FOMKP considered in this paper. Particularly, the FOMKP relaxes the above problems by allowing the amount of items allocated to each knapsack to be a continuous variable. This allows us to achieve competitive ratios that are independent of the dimension of knapsacks.

\textit{The One-Way Trading Problem.} This problem was first introduced and studied by EI-Yaniv \textit{et al.}~\cite{El2001} under the assumption that the price is bounded from above and below. It is shown that a threat-based algorithm can achieve a competitive ratio of $O(\ln\theta)$, where $\theta$ is the ratio of upper and lower bounds on the prices (or conversion rates). Follow-up works mainly focus on variants with different assumptions on the prices, e.g., known distribution of prices~\cite{fujiwara2011average}, unbounded prices~\cite{chin2015competitive}, and interrelated prices~\cite{schroeder2018optimal}, or with different performance metrics, e.g., competitive difference~\cite{wang2016competitive}. Recently, Yang \textit{et al.}~\cite{yang2019online} considers a bounded price but a different problem setting, in which the investor is unaware of whether a price is the last one and may have leftover assets due to the sudden termination of trading process. This work designs a threshold-based algorithm that is shown to be $(1+\ln\theta)$-competitive. A further extension was given by Lin \textit{et al.}~\cite{Lin2019}, which generalizes the linear objective function to a concave one. In this paper, we consider a problem that generalizes all these variants, and we present a single algorithm that matches or improves upon the competitive ratio in each case. 

\textit{The Online EV Charging Problem.} The task of managing the charging of EVs is a prominent algorithm challenge for smart energy systems \cite{Guo2017optimal,Lee2018large,Alinia2020online}. A number of variants of online EV charging have been tackled in prior work~\cite{Alinia2020online,alinia2018competitive,alinia2019online,sun2020orc,zheng2014online}. Despite the fact that many algorithms achieve good performances (on average) in practice, analyzing algorithms to provide worst-case guarantees for online EV charging is notoriously difficult, and existing algorithms, such as model predictive control (MPC), are known to be vulnerable to adversarial inputs. For example, in~\cite{Guo2017optimal}, the EV charging problem is modeled as an online linear program and the authors show that MPC is equivalent to an offline solver when the costs are uniformly monotone and has a competitive ratio $O(\theta)$ otherwise.  Most commonly, online EV algorithms have been allowed to adaptively determine the EV charging schedule over time, thus providing no guarantees to a driver at arrival about the total charge they will receive, e.g., \cite{Guo2017optimal,Lee2018large,Alinia2020online}.  This approach simplifies the analysis; however, charging with on-arrival commitment~\cite{alinia2019online,sun2020orc} is what is desired by drivers. In this paper, we consider charging with on-arrival commitment, in which the charging schedule is determined upon EVs' arrivals and will be kept unchanged. Competitive analysis of this setting is known to be challenging, e.g., \cite{alinia2019online} has shown that no bounded competitive ratio can be achieved in general. However, in this paper, we give an online algorithm with a nearly optimal competitive ratio under a set of regularity conditions that are standard in the online knapsack literature (Assumption~\ref{ass:value-function}). Further, our algorithm achieves its competitive ratio when charging rate constraints are included, which add additional challenges and are typically not considered in online EV charging formulations.

\section{The Online Fractional Multiple Knapsack Problem}

This paper focuses on a novel generalization of the fractional Online Multiple Knapsack Problem (\fomkp). In a classical \omkp, each arriving item can only be packed into one of the knapsacks.  In contrast, in the \fomkp each knapsack $m\in\calm$ is allowed to accept a fraction of the entire size of each item $n\in\caln$, i.e., the accepted item can be packed into multiple knapsacks, each of which accommodates a portion of the total accepted item. Additionally, the formulation we consider incorporates heterogeneous rate-limiting constraints depending on the knapsack and the item to be packed.  This generalization is motivated by issues in practical problems such as online EV charging and enables the unification of a wide range of classical online algorithms problems, many of which are traditionally approached with contrasting algorithmic techniques. 

\subsection{Problem Statement}
\label{sec: problem}
We consider a setting where items in a set $\caln$ need to be packed into  knapsacks in a set $\calm$. For each item $n$, the operator decides an \textit{assignment vector} denoted by $\by_n:=(y_{n1},\dots,y_{nM})$, where each entry $y_{nm}$ is the fraction of item $n$ packed into the knapsack $m$. The assignment vector $\by_n$ must satisfy the following constraints. The set of assignment vectors $\by_n$ satisfying~\eqref{p:1}-\eqref{p:3} is  $\caly_n$:
\begin{align}
	\label{p:1}
	\sum\nolimits_{m\in\calm} y_{nm} \le  D_n, \quad &  \forall n\in\caln,\\
	\label{p:2}
	\sum\nolimits_{n\in\caln}y_{nm} \le C_m, \quad  & \forall m\in\calm,\\
	\label{p:3}
	0 \le y_{nm}\le  Y_{nm}, \quad  & \forall n\in\caln,  m\in\calm.
\end{align}

The first constraint~\eqref{p:1} is a \textit{demand constraint}, which bounds the total accepted fractions of the item $n$ by the item size $D_n$. The second constraint \eqref{p:2} ensures the assigned fractions $\by_n$ satisfy the \textit{capacity constraints}, of the heterogeneous knapsacks, where $C_m$ is the maximum capacity of the knapsack $m$. The third constraint~\eqref{p:3} is a \emph{rate constraint}, which ensures that at most $Y_{nm}$ fraction of the item $n$ can be packed into the knapsack $m$.  This constraint also allows imposing heterogeneous restrictions on which items can be packed into which knapsacks, e.g., by setting $Y_{nm}= 0$ for knapsacks that are not available to the item $n$.  Due to the algorithmic difficulties it creates, the rate constraint~\eqref{p:3} is rarely studied in the literature of \omkp. Note that all three of these constraints are crucial to capturing applications such as EV scheduling.  We highlight this in Section \ref{subsec:examples}.

\subsubsection{Objective Function}  The objective of an \fomkp is to optimize the value of packed items subject to the constraints \eqref{p:1}-\eqref{p:3}.  More formally, let ${g_n(\by_n): \caly_n \to \mathbb{R}^+}$ denote the \textit{value function} of the item $n$. This function models the value of the item $n$ with an assignment vector $\by_n$.
Optimizing over assignment vectors that satisfy~\eqref{p:1}-\eqref{p:3}, the offline version of \fomkp can be summarized as
\begin{align}
	\label{p:omkp}
	({\rm Offline}\ \fomkp)\quad\max_{\by_n}\quad  \sum\nolimits_{n\in\caln}  g_n(\by_n),
	\quad \text{ s.t. }\quad& {\rm constraints\ } \eqref{p:1}-\eqref{p:3}.
\end{align}

In this paper, we follow standard practice in the literature and focus on value functions that are separable  or aggregate functions, e.g.,  \cite{qureshi2009cutting,adler2011algorithms,Alinia2020online}. These definitions are useful in order to prove competitive bounds. 

\begin{dfn}[Aggregate Function]\label{ass:aggregate-value-function}
	The aggregation of allocations contributes to the value function, i.e., ${g_n(\by_n) = g_{n}(\sum_{m\in\calm}y_{nm})}$.
\end{dfn}

\begin{dfn}[Separable Function]\label{ass:separable-value-function}
	The value function is separable over allocations, i.e.,\\ ${g_n(\by_n) = \sum_{m\in\calm} g_{nm}(y_{nm})}$, where $g_{nm}(y_{nm})$ is the value of allocating $y_{nm}$ of item $n$ to knapsack $m$.
\end{dfn}

Both definitions capture a broad range of applications. For example, interpreting different knapsacks as different time slots allows us to model the EV charging application (described in detail in Section \ref{subsec:examples}) using an aggregate function. Additionally, the notion of separable value functions captures the phenomenon of different values for allocating items to different knapsacks, which  is of interest to applications such as geographical load balancing (described in detail in Section \ref{subsec:examples}). Note that when the value function is linear, an aggregate function is by definition a separable function. Additionally, when there is only one knapsack, both definitions are equivalent.

In addition, we assume that the value functions satisfy the following regularity conditions.
\begin{ass}\label{ass:value-function}
	The value functions $\{g_n:n\in\caln\}$ satisfy:
	
	(i) for any $n\in\caln$, $g_n(\cdot)$ is non-decreasing, differentiable and concave in $\caly_n$; 
	
	(ii) for any $n\in\caln$, $g_n(\boldsymbol{0}) = 0$; 
	
	(iii) the partial derivative of $g_n(\cdot)$ is bounded, i.e., there exist constants  $L,U>0$ such that for any $n\in\caln$ and $m\in\calm$,  $L \le \frac{\partial g_n}{\partial y_{nm}} \le U$.
\end{ass}

These are again classical assumptions in the online knapsack literature~\cite{zhou2008budget,Zhang2017,Lin2019,tan2020}. The first condition ensures that the value function is smooth and has diminishing returns. The second condition indicates that packing no item earns no value. The third condition requires that the partial derivatives of the value function are lower and upper bounded by $L$ and $U$, respectively. $L$ and $U$ are assumed to be known and let $\theta := U/L$ denote the fluctuation ratio.

\subsubsection{Online Formulation}

The parameters described to this point can be encapsulated in two sets, $\cals$ and $\cali$. 
The set $\cals:=\{\{C_m\}_{m\in\calm}, L, U\}$ includes the capacity information, and the partial derivative bounds of value functions. We call $\cals$ the setup information since it is known from the start and can be used for the design of online algorithms.
The set $\cali: = \{D_n, \{Y_{nm}\}_{m\in\calm}, g_n(\cdot) \}_{n\in\caln}$ contains the information corresponding to each item, including the item size, rate limits, and value functions. 
$\cali$ is also called arrival information.
The focus of this paper is an online formulation where the arrival information of each item is revealed upon its arrival. Thus, the algorithm only knows the causal information $\{D_k, \{Y_{km}\}_{m\in\calm}, g_k(\cdot)\}_{k=1,\dots,n}$ for the decision-making of item $n$. 

Our goal is to design an online algorithm that makes an irrevocable assignment decision based only on causal information and still performs nearly as well as the offline optimum. Particularly, we evaluate the performance of an online algorithm under the competitive analysis framework. Given setup information $\cals$, let $\opt(\cali)$ and $\alg(\cali,\cala)$ denote the offline optimum of the \fomkp and the value achieved by an online algorithm $\cala$ under an arrival instance $\cali$, respectively. The competitive ratio of the online algorithm $\cala$ is defined as
$
\texttt{CR}(\cala) = \max_{\cali \in\Omega} \frac{\opt(\cali)}{\alg(\cali,\cala)},
$
where $\Omega$ denotes the set of all instances that satisfy Assumption~\ref{ass:value-function}. An algorithm $\cala$ is  $\alpha$-competitive if $\texttt{CR}(\cala)\leq \alpha$.

\subsection{Examples} \label{subsec:examples}

The generalization of \fomkp introduced above is novel and serves to unify a wide variety of classical online problems. On one hand, it is a generalization of two classical online optimization problems: the one-way trading and the online knapsack problems, bringing together two streams of research that were previously treated separately in the literature. On the other hand, it is the core model of many practical online decision-making applications such as the EV charging and online geographical load balancing problems. We highlight these connections explicitly in the following. Previous studies on those problems are summarized at the end of Section~\ref{sec:intro} 

{\bf The One-Way Trading Problem.} 
One-way trading~\cite{El2001} is a classical online problem where an investor aims to trade a limited amount of one asset (e.g., dollar) to another asset (e.g., yen). The sequence of the trading prices is not known to the investor ahead of time and the investor must decide the amount of traded assets for each price except the last one, and trade remaining assets at the last price. The objective is to maximize the total profits of the entire trading process.

To see that one-way trading is a special case of \fomkp, observe that \fomkp reduces to the following generalized one-way trading (\got) problem when the number of knapsacks is $M = 1$ and the rate limit is equal to the item size, i.e., $Y_{n1} = D_n$,
\begin{align}\label{p:got}
	({\rm Offline}\ {\got})\quad\max_{y_n} \quad \sum\nolimits_{n\in\caln} g_n(y_n), \quad {\rm s.t.}\ \sum\nolimits_{n\in\caln}y_n \le C,\quad
	0\le y_n\le D_n, \quad\forall n\in\caln,
\end{align}
where we omit the knapsack index for simplicity. Notice that  
\got includes all previous variations of the one-way trading problem in the literature, e.g.,~\cite{yang2019online,El2001}, and additionally extends the most general one-way trading model in~\cite{Lin2019} by including the rate constraint $y_n\le D_n$.

{\bf The Online Multiple Knapsack Problem with Small Weights.} While our focus is on the \fomkp, there are strong connections between the fractional and integral versions. In the integral version, items must be assigned to a single knapsack and cannot be split between multiple knapsacks. Thus, \omkp~\cite{zhou2008budget} is an online integer linear program with multiple capacity constraints. 

To see the connection between \fomkp and \omkp, note that, when the assignment set $\caly_n$ is restricted to the following discrete set  
\begin{align}\label{eq:dfn-discrete-set}
	\tilde{\caly}_n := \left\{\by_n: \sum\nolimits_{m\in\calm} y_{nm} \le D_n, y_{nm}\in \{0,Y_{nm}\}, \forall m\in\calm \right\}, 
\end{align}
where $Y_{nm} \in \{0, D_n\}$, \fomkp becomes an \omkp with assignment restrictions (\omkpar)~\cite{Kellerer2004}, which is a generalization of \omkp.
Under the assumptions that (i) each item can be packed into any one of all knapsacks, i.e., $Y_{nm} = D_n, \forall m\in\calm$, and (ii) the value of each item is independent of knapsacks, e.g., $g_n(\by_n)$ is aggregate, \omkpar reduces to \omkp.

Under an infinitesimal assumption that is standard in the literature (i.e., the weights/sizes of items are much smaller than the knapsack capacities)~\cite{zhou2008budget,tan2020}, the online algorithms designed for the \fomkp can be converted to an integral version to solve the \omkpar with the same competitiveness. This is formally highlighted in Remark~\ref{rem:ota-integral}. 
Thus, the \fomkp can be considered as a generalization of \omkpar under an infinitesimal assumption. Note that this assumption is typically satisfied in practical applications. 
For example, the energy demand required by a single EV is much smaller than the capacity of the garage and the resource required by a single job or VM is much smaller than the capacity of servers. Additionally, such an infinitesimal assumption is typically required for online (non-fractional) knapsack problems, e.g., \cite{zhou2008budget,Zhang2017,tan2020}, since no non-trivial competitive results are known without assumptions~\cite{zhou2008budget}. 

{\bf Online EV Charging with On-Arrival Commitment.} 
In an online EV charging problem, an operator of an EV charging facility charges a set $\caln$ of EVs that arrive sequentially in a set $\calm$ of time slots. 
The facility capacity (i.e., the available total charging power) at time $m$ is $C_m$. Each EV $n$ is characterized by parameters $\{\calm_n, D_n, Y_{n}, g_n(\cdot)\}$, where $\calm_n:=\{t^a_n,\dots,t^d_n\}$ denotes the available window of EV $n$ with $t_n^a$ and $t_{n}^d$ as the arrival and departure times, $D_n$ denotes the charging demand, $Y_{n}$ is the charging rate limit, and $g_n(\cdot)$ is a value function of the total received energy of EV $n$. By letting $Y_{nm} = Y_n, \forall m\in\calm_n$, and $Y_{nm} = 0, \forall m\in\calm\setminus\calm_n$, the offline EV charging problem is precisely the problem~\eqref{p:omkp}.
In order to ensure the quality of service (QoS) for drivers in the online setting, the operator needs to guarantee how much energy each EV will be provided when the EV first plugs in, i.e., on-arrival commitment.
Crucially, we consider a specific approach to achieve such on-arrival commitment by committing to a charging schedule upon the arrival of each EV. 
Particularly, upon the arrival of each EV $n$, the operator commits to a charging schedule $\by_n = (y_{n1},\dots,y_{nM})$, where $y_{nm}$ is the charging rate of EV $n$ at time $m$, and obtains a value $g_n(\by_n):=g_n(\sum_{m\in\calm}y_{nm})$ from charging EV $n$ by the schedule $\by_n$. The goal of the operator is to maximize the total values of all EVs. We can then see the online EV charging with on-arrival commitment exactly fits into \fomkp.

Note that our formulation considers the charging schedules as irrevocable decisions, i.e., $\by_n$ will be fixed on arrival of EV $n$. 
This formulation restricts the decision space of online EV charging since the operator has the flexibility of re-optimizing the charging schedules over time in practice.
However, to the best of our knowledge, no online EV charging algorithms with re-optimizing in the literature have achieved bounded competitive ratios and on-arrival commitment simultaneously. 
As is shown in next section, under our formulation with the most intuitive way to guarantee on-arrival commitment, we can design online algorithms that can achieve nearly optimal competitive ratios within an additive factor of one. Such results indicate our formulation  does not lose too much flexibility. 

Moreover, compared to \fomkp, the online EV charging, which models time slots as knapsacks, actually has additional information, i.e., past time slots will not be available for assignment of future arrivals. 
Thus, the operator can adjust the charging rates slot by slot based on the committed charging schedules to further improve the competitive ratios. 
Particularly, if the committed schedules only consume partial capacity of the current time slot, the operator can allocate the remaining capacity to the EVs whose demand constraints and rate-limiting constraints of current time slot are not binding.
Since value functions of all EVs are non-decreasing in their total received energy, the resulting schedules after such adaptive adjustment can achieve the total value no worse than the committed schedules, and hence improve the empirical competitive ratios in practice.   

\textbf{Cloud Scheduling and Geographical Load Balancing.} \fomkp with separable value functions can be viewed as an extension of the classical task of job scheduling in a cloud, e.g., ~\cite{lucier2013efficient,zheng2016online}, which includes the so-called geographical load balancing problem, e.g.,~\cite{qureshi2009cutting,liu2014greening,adnan2012energy,luo2013temporal}. More concretely, consider a service provider, e.g., Netflix or YouTube, with a set of geographically distributed  infrastructures for performing video processing jobs, e.g., encoding video files into multiple quality levels~\cite{wang2019sustainable} to be used in ABR streaming algorithms~\cite{spiteri2020bola}. These video processing jobs are heterogeneous and typically have flexibility in execution across different locations without violating QoS requirements.  

While the problems of job scheduling and geographical load balancing have been mainly studied separately, with \fomkp, the joint problem could be tackled.
Formally, in this model a knapsack $m\in\calm$ denotes a pair of time and location, i.e., with $T$ time slots and $K$ locations, we have $T\times K = M$ knapsacks. There is a set of $\caln$ (e.g., video) jobs, each characterized by $\{\calm_n, D_n, Y_n, \{g_{nm}(\cdot)\}_{m\in\calm_n}\}$, where $\calm_n$ is the set of slot/locations available for job $n$, $D_n$ is the computation demand, and $Y_n$ is the job parallelism bound~\cite{lucier2013efficient} that captures the maximum number of processing units (or servers) that can be allocated to a single job at any given slot/location. Lastly, $g_{nm}(\cdot)$ captures the value (or cost, e.g., energy~\cite{qureshi2009cutting} or bandwidth~\cite{adler2011algorithms}) of executing job $n$ at slot/location $m$. In this model, the deadline constraints and QoS constraints, e.g., infeasibility of running jobs in far locations~\cite{gupta2019combining} could be captured using the rate constraints. In particular, by letting $Y_{nm} = Y_n, \forall m\in\calm_n$, and $Y_{nm} = 0, \forall m\in\calm\setminus\calm_n$, the deadline and QoS constraints of job $n$ could be enforced. 

\section{Algorithms \& Results}
The key challenge when designing online algorithms for \fomkp results from the capacity constraints that couple the knapsack decisions of all items. Formally, one can understand the difficulty created by this via the dual variables.  In particular, if the optimal dual variables associated with the capacity constraints were to be given, the \fomkp could be decoupled across items and the optimal knapsack decision for each item could be determined by maximizing a pseudo-utility that is defined as the value of the item minus a linear cost using the optimal dual variables as the price.
However, in the online setting, the optimal dual variables cannot be known since the future items' information is unavailable. 
Thus, we can only use an adaptive estimation of the dual variables to solve the online problem based on causal information.

This intuition leads to an important algorithmic idea at the core of literature focusing on the \omkp, e.g., \cite{zhou2008budget}: \textit{estimate the dual variable as a function of the knapsack utilization}, i.e., the fraction of the consumed knapsack capacity. We refer to this estimation function as a \textit{threshold function} defined below.

\begin{dfn}[Threshold Function]
	A threshold function $\phi_m(w)$ of a knapsack $m$ is a non-decreasing function that evaluates the price (or marginal cost) of packing items to the knapsack $m$ when its utilization $w$ is within capacity $w\in[0,C_m]$ and $\phi_m(w) = +\infty$ when $w\in(C_m,+\infty)$.
\end{dfn}

The approach we follow in this paper is to design a class of online threshold-based algorithms (\ota) for \fomkp. In the following, we first formally introduce the \ota class of algorithms, followed by key ideas for competitive analysis, and finally we present our main competitive results. Proofs are deferred to the next sections.

\subsection{Online Threshold-Based Algorithms (\ota)}
\label{subsec:ota}

The \ota framework has been developed in the context of \omkp
by Zhou \textit{et al.}~\cite{zhou2008budget}. The basic idea of \ota is to use threshold functions to estimate the cost of a (non-fractional) knapsack assignment under infinitesimal assumptions and determine the online solution by solving a pseudo-utility maximization problem, i.e., the value from the item minus the cost of packing it. 
We extend this idea to \fomkp, where the estimated cost of assignment decisions is estimated by an integral of the threshold function.

More formally, given a set of threshold functions $\phi:=\{\phi_m(\cdot)\}_{m\in\calm}$, the details of the \ota algorithm are provided in Algorithm~\ref{alg:ota-overall}.
Let $\by_n^*$ be the online assignment decision produced by $\ota_\phi$.
Let $\smash{w_m^{(n)} = \sum_{k=1}^{n-1} y_{km}^*}$ denote the utilization of a knapsack $m$ observed upon the arrival of item $n$, which is the total fraction of the occupied knapsack capacity by the previous $n-1$ items. 
$\ota_\phi$ uses the utilization as the state for decision-making.
Since $\phi_m(u)du$ can be considered as the cost of assigning a small bit of the item to knapsack $m$ when its current utilization is $u$, we can estimate the total cost of assigning $y_{nm}$ fraction of the item $n$ to the knapsack $m$ by an integral ${\smallint_{w_m^{(n)}}^{w_m^{(n)}+ y_{nm}}\phi_m(u)du}$. Therefore, the second term of the pseudo-utility in the problem~\eqref{p:utility-maximization} is the total estimated cost of a knapsack assignment $\by_n$.
Since $\phi$ is a non-decreasing function, the estimated cost is a convex function in $\by_n$ and this pseudo-utility maximization problem~\eqref{p:utility-maximization} can be efficiently solved. By definition, the threshold function becomes infinite when the utilization exceeds the capacity, which avoids violating the knapsack capacities.

The above highlights that $\ota_\phi$ is fully parameterized by the threshold function $\phi$. Thus, the key design question is how to determine the threshold function $\phi$ such that $\ota_\phi$ is competitive with the offline optimum. Interestingly, prior works, e.g., \cite{yang2019online,zhou2008budget,Zhang2017}, use the same threshold function for the classical one-way trading and online $0/1$ knapsack problems. However, this threshold function is obtained through trial and error, and it is unclear how to design  threshold functions for more complicated variations or other settings. The crucial bottleneck for progress of these algorithms is understanding how to design the threshold function, and the key idea in our work is a systematic approach for the design of such threshold functions, which we describe in the next section. 

\begin{rem}\label{rem:ota-integral}
	Our focus is on fractional knapsack problems, but $\ota_\phi$ can be easily converted into an integral version for solving the non-fractional problem \omkpar, which restricts the schedule to a discrete set $\tilde{\caly}_n$ defined in equation~\eqref{eq:dfn-discrete-set}~\cite{Kellerer2004}.
	To do so, the estimated cost of packing item $n$ to knapsack $m$ is approximated by ${\smallint\nolimits_{w_{m}^{(n)}}^{w_{m}^{(n)}+ y_{nm}}}\phi_m(u)du \approx \phi_m(w_{m}^{(n)}+ y_{nm}) y_{nm}$.  Under the infinitesimal assumption, this approximation is accurate enough and the integral \ota can achieve the same competitive ratio for \omkpar as that of \ota for \fomkp. 
\end{rem}

\begin{algorithm}[t]
	\caption{Online Threshold-Based Algorithm with Threshold Function $\phi$ ($\ota_\phi$)}
	\label{alg:ota-overall}
	\begin{algorithmic}[1]
		\State \textbf{input:} threshold function $\phi :=\{\phi_m(\cdot)\}_{m\in\calm}$, and initial knapsack utilization $w_{m}^{(1)} = 0, \forall m\in\calm$;
		\While{item $n$ arrives}
		\State observe item size $D_n$, rate limits $\{Y_{nm}\}_{m\in\calm}$, and value function $g_n(\cdot)$;
		\State determine knapsack assignment $\by_n^*$ by solving the pseudo-utility maximization problem
		\begin{align}\label{p:utility-maximization}
			\by_n^* = \argmax_{\by_n\in\caly_n} \quad g_n(\by_n) - \sum\nolimits_{m\in\calm}\smallint\nolimits_{w_{m}^{(n)}}^{w_{m}^{(n)}+ y_{nm}}\phi_m(u)du;
		\end{align}
		\State update the utilization $w_{m}^{(n+1)} = w_{m}^{(n)} + y_{nm}^*, \forall m\in\calm$.
		\EndWhile
	\end{algorithmic}
\end{algorithm}

\subsection{Key Idea: Designing the Threshold Function via Instance-dependent Online Primal-dual Analysis }
\label{subsec:approach}

The fundamental challenge when developing an $\ota_\phi$ algorithm is the design of the threshold function $\phi$.  The key idea of the approach proposed in this paper is to design $\phi$ using an instance-dependent primal-dual analysis that extracts the design of the threshold function from the identification of a worst-case instance.  

The use of online primal-dual (\opd) analysis for \ota stems from the work of~\cite{buchbinder2009design}.
The key idea of the \opd approach is to construct a feasible dual solution based on the online solution produced by the online algorithm to be analyzed, and then build the upper bound of the offline optimum using the feasible dual objective based on weak duality~\cite{boyd2004convex}. More concretely, since $\ota_{\phi}$ is only parameterized by the threshold function $\phi$, the performance of $\ota_{\phi}$ can be denoted by $\alg(\cali,\phi)$. 
Let $\dual(\cali,\phi)$ denote the objective of the dual problem of \fomkp evaluated at the constructed dual solution. Therefore, $\dual(\cali,\phi)$ is also a function of $\phi$. This means that the \opd technique allows the design of $\phi$ to be viewed as a search for $\phi$ such that:
\begin{align}\label{eq:opd}
	\alpha \alg(\cali,\phi) \ge  \dual(\cali,\phi) \ge \opt(\cali), \forall \cali\in\Omega.
\end{align}
The first inequality holds only under certain sufficient conditions, i.e.,  $\phi$ must satisfy a set of differential equations parameterized by $\alpha$, see \cite{Devanur2012,tan2020,Zhang2017} for examples.
The second inequality comes from weak duality and holds if the constructed dual solution is feasible. 
The classical \opd mainly focuses on designing the sufficient conditions to reduce the gap in the first inequality but neglects the possibility of reducing the weak duality gap in the second inequality, which also makes a difference to derive a smaller competitive ratio $\alpha$.

The classical approach for designing such a $\phi$ in the literature (e.g.,~\cite{buchbinder2009design,Devanur2012,buchbinder2009online,tan2020}) uses the primal-dual relationship (e.g., weak duality, KKT conditions) between an offline primal problem and its dual. This viewpoint does not rely on understanding instances of particular online optimization problems.
However, the gap between $\opt(\cali)$ and $\dual(\cali,\phi)$ not only depends on the constructed dual solution, but also the constraint coefficients of the primal problem. Thus, under different instances, the dual objective based on the primal constraints in the same offline formulation may lead to a loose upper bound. 

The novelty of our approach is the construction of instance-dependent offline formulations by adding constraints to the primal problem that are constructed based on online solutions, and then utilizing the corresponding dual objectives to bound the offline optimum. In this way, we actually perform an instance-dependent \opd analysis. Moreover, by focusing on the worst-case instances, this approach yields threshold functions that are tuned to the challenges of the online problem, and are tight for the worst case.  

While the application of this approach is complex for the general case of \fomkp, it can be illustrated concretely in the specific case of the generalized one-way trading (\got) described in the problem \eqref{p:got}. In that setting, the following lemmas (see details in Section~\ref{sec:got}) provide a simple, concrete illustration of the approach. First, Lemma~\ref{lem:sufficient} provides a sufficient condition on $\phi$ that can ensure $\ota_\phi$ is $\alpha$-competitive. 

\begin{lem}
	\label{lem:sufficient}
	Under Assumption~\ref{ass:value-function}, $\ota_\phi$ for \got is $\alpha$-competitive if the threshold function $\phi$ is
	\begin{align*}
		\phi(w)=
		\begin{cases}
			L&  w\in[0,\beta)\\
			\varphi(w) & w\in[\beta,C]
		\end{cases},
	\end{align*}
	where $\beta \in [0,C]$ is a utilization threshold and $\varphi$ is a non-decreasing function, and $\phi$ satisfies
	\begin{align}\label{eq:ode-sufficient}
		\begin{cases}
			\varphi(w) C \le \alpha \smallint_{0}^{w} \phi(u)du, w\in[\beta, C], \\
			\varphi(\beta) = L, \varphi(C) \ge U. 
		\end{cases}
	\end{align} 
\end{lem}

The form of the threshold function specified by the lemma consists of two segments, a flat segment in $[0,\beta)$ and a non-decreasing segment in $[\beta,C]$.
This two-segment function results from two families of instances as shown in Case I and Case II of Section~\ref{subsec:proof-threshold-got}, in which different offline formulations are needed to construct the dual objective $\dual(\cali,\phi)$ such that the gap between $\dual(\cali,\phi)$ and  $\opt(\cali)$ is minimized.
The differential equations and boundary conditions in~\eqref{eq:ode-sufficient} are designed to guarantee the first inequality of the \opd relationship~\eqref{eq:opd} holds.
By binding all inequalities and solving equations~\eqref{eq:ode-sufficient}, the resulting threshold function $\phi^*$ achieves the minimal competitive ratio among the threshold functions that satisfy this sufficient condition.   
This competitive ratio is an upper bound of the optimal competitive ratio and its tightness depends on the instances.

Conversely, the next key lemma shows necessary conditions that need to be satisfied in order to  achieve $\alpha$-competitiveness. It is phrased in terms of a \emph{utilization function}, which is an abstracted model of an online algorithm, mapping an instance to the \textit{final} utilization level of the knapsack.  See Definition~\ref{dfn:instance} for a formal definition.

\begin{lem}
	\label{lem:necessary-condition}
	If there exists an $\alpha$-competitive online algorithm for \got, there must exist a utilization function $\psi(p): [L,U]\to[0,C]$ such that $\psi$ is a non-decreasing function and satisfies 
	\begin{align}\label{eq:necessary-condition}
		\begin{cases}
			L\psi(L)+\smallint_{L}^{p} u d\psi(u) \ge pC/\alpha, p\in[L,U],\\
			\psi(L) \ge C/\alpha, \psi(U) \le C.
		\end{cases}
	\end{align} 
\end{lem}

This lemma provides an interpretation of an online algorithm for \got as a black box, with an instance as an input and a sequence of changes in the knapsack utilization as an output. 
Given a family of instances, each online algorithm corresponds to a utilization function $\psi$, and in Lemma~\ref{lem:necessary-condition} we specifically design the family of instances in a way that allows them to be indexed by a continuous marginal value within $[L,U]$ (see Definition~\ref{dfn:instance}), making $\psi$ a simple single variable function.
If $\psi$ satisfies \eqref{eq:necessary-condition}, the online algorithm corresponding to $\psi$ can achieve at least $1/\alpha$ of offline optimum under the specifically-designed family of instances. 
Thus, the existence of a solution to \eqref{eq:necessary-condition} is necessary for the existence of an $\alpha$-competitive online algorithm.	
Moreover, Lemma~~\ref{lem:necessary-condition} means the $\alpha$ in \eqref{eq:necessary-condition} is a lower bound of the optimal competitive ratio if it ensures there exists a solution to \eqref{eq:necessary-condition}. The minimal $\alpha$ can be achieved when all inequalities are binding in differential equations~\eqref{eq:necessary-condition} and let $\psi^*$ denote the corresponding solution.  
An important observation is that the two sets of differential equations~\eqref{eq:ode-sufficient} and~\eqref{eq:necessary-condition} with the same $\alpha$ are essentially the same when all inequalities are binding.  In particular, the threshold function $\phi^*$ is an inverse function of the utilization function $\psi^*$. 
This implies that $\ota_{\phi^*}$ achieves not only an upper bound but also a lower bound, and hence exactly the optimal competitive ratio of \got. Conversely, it also means the special instance used for constructing the necessary condition is actually the worst-case instance of \got. Thus, the $\ota_{\phi^*}$ algorithm and the worst-case instances are connected via the differential equations, implying the design of $\ota_{\phi^*}$ is equivalent to finding the worst-case instance for \got.

Returning to the general \fomkp, the instance-dependent \opd approach can still be leveraged to design competitive \ota algorithms via an understanding of the worst-case instance, though the application is more complex (see Section~\ref{sec:omkp}). 
The optimal \ota and the worst-case instance can be derived simultaneously when the upper and lower bounds match.

\subsection{Summary of Results}
\label{subsuc:summary-of-results}

Using the instance-dependent \opd approach described in the previous section, the main results of the paper present the design of threshold functions for \fomkp and the special case of \got, which has received considerable attention in the literature. 

For simplicity, we start by discussing the special case of \got, where we design a threshold function that achieves the optimal competitive ratio $1 + \ln\theta$.

\begin{thm}\label{thm:threshold-got}
	Under Assumption~\ref{ass:value-function}, when the threshold function of $\ota_\phi$ for \got is
	\begin{align}\label{eq:threshold-got}
		\phi^*(w)=
		\begin{cases}
			L & w\in[0,\beta^*)\\
			L e^{(1+\ln\theta)w/C - 1} & w\in[\beta^*, C]
		\end{cases},
	\end{align}
	where $\beta^* = \frac{C}{\alpha_{\phi^*}}$ is the utilization threshold, the competitive ratio of $\ota_{\phi^*}$ is $\alpha_{\phi^*} = 1+\ln\theta$.
\end{thm}

\begin{thm}\label{thm:opt-CR}
	The optimal competitive ratio of \got is $\alpha^* = 1 + \ln\theta$.
\end{thm}

This is the best-known result for \got (under Assumption~\ref{ass:value-function}), improving upon the results summarized in the one-way trading problem~\cite{yang2019online} and the online $0/1$ knapsack problem~\cite{Zhang2017,zhou2008budget}.
Importantly, the threshold function $\phi^*$ in~\eqref{eq:threshold-got} coincides with the optimal threshold function of \ota used in~\cite{yang2019online,Zhang2017,zhou2008budget}, and $\ota_{\phi^*}$ achieves the same competitive ratio.
Thus, Theorem~\ref{thm:threshold-got} highlights that generalizing the objective of the one-way trading problem from linear to concave functions and taking into account the rate limits do not degrade the competitive performance of \ota.	

In addition, our approach can be extended to solve two previously studied variants of \got. First, \cite{El2001} considers {a setting that the investor can trade all its remaining assets at the last (lowest) price. 
	We consider a variant of the \got defined in~\eqref{p:got} that also allows to fill the remaining knapsack capacity with items of the lowest marginal value $L$.   
	In this context our approach yields a threshold function for \ota that achieves the same optimal competitive ratio as the special case in~\cite{El2001}.} Second, \cite{Lin2019} considers a relaxed assumption on the value function, which restricts the average value of each item to be lower bounded by $L/c$, instead of the marginal value in \got, with a given parameter $c\ge 1$ (see Assumption~\ref{ass:relaxed} for detail). 
In this context our approach yields a threshold function for \ota that improves the upper bound on the competitive ratio from $O(c(1+ \ln\theta))$ in~\cite{Lin2019} to $O(\ln(c\theta))$. Beyond these cases, our result also applies to \got problems for which no previous bounds on the competitive ratio were known.

Obtaining results for the general \fomkp problem is more challenging than in the \got setting; however, the same approach we introduce in the \got setting can be generalized.  For the general case, using the instance-dependent \opd approach, we design a threshold function that nearly achieves the optimal competitive ratio -- it differs by an additive factor of one.  In this case, we have two results, one for the case of aggregate value functions (see Definition~\ref{ass:aggregate-value-function}) and one for the case of separable value functions (see Definition~\ref{ass:separable-value-function}).  
\begin{thm}\label{thm:threshold-omkp-aggregate}
	Under Assumptions~\ref{ass:value-function}, when the threshold function of $\ota_\phi$ for the \fomkp with an aggregate value function is
	\begin{align}\label{eq:threshold-fomkp-aggregate}
		\phi_m^*(w)=
		\begin{cases}
			L & w\in[0,\beta_m^*)\\
			Le^{\frac{\alpha_{\phi^*} }{C_m}w - \frac{\alpha_{\phi^*}}{\alpha_{\phi^*} - 1}} & w\in[\beta_m^*, C_m]
		\end{cases},
	\end{align}
	where $\beta_m^* = \frac{C_m}{\alpha_{\phi^*} - 1}$, the competitive ratio of $\ota_{\phi^*}$ is the solution of $\alpha_{\phi^*} - 1 - \frac{1}{\alpha_{\phi^*} - 1} = \ln\theta$.
\end{thm}

\begin{thm}\label{thm:threshold-omkp-separable}
	Under Assumption~\ref{ass:value-function}, when the threshold function of $\ota_\phi$ for the \fomkp with a separable value function is
	\begin{align}\label{eq:threshold-fomkp-separable}
		\phi_m^*(w)=
		\begin{cases}
			L & w\in[0,\beta_m^*)\\
			\frac{U-L}{e^{\alpha_{\phi^*}} - e^{{\alpha_{\phi^*}}/{(\alpha_{\phi^*} - 1)}}} e^{\frac{\alpha_{\phi^*}}{C_m} w} + \frac{L}{\alpha_{\phi^*}} & w\in[\beta_m^*, C_m]
		\end{cases},
	\end{align}
	where $\beta_m^* = \frac{C_m}{\alpha_{\phi^*} - 1}$, the competitive ratio of $\ota_{\phi^*}$ is the solution of $\alpha_{\phi^*} - 1 - \frac{1}{\alpha_{\phi^*} - 1} = \ln\frac{\alpha_{\phi^*} \theta - 1}{\alpha_{\phi^*} - 1}$.
\end{thm}
\begin{figure}[t]
	\center
	\includegraphics[width=0.45\textwidth]{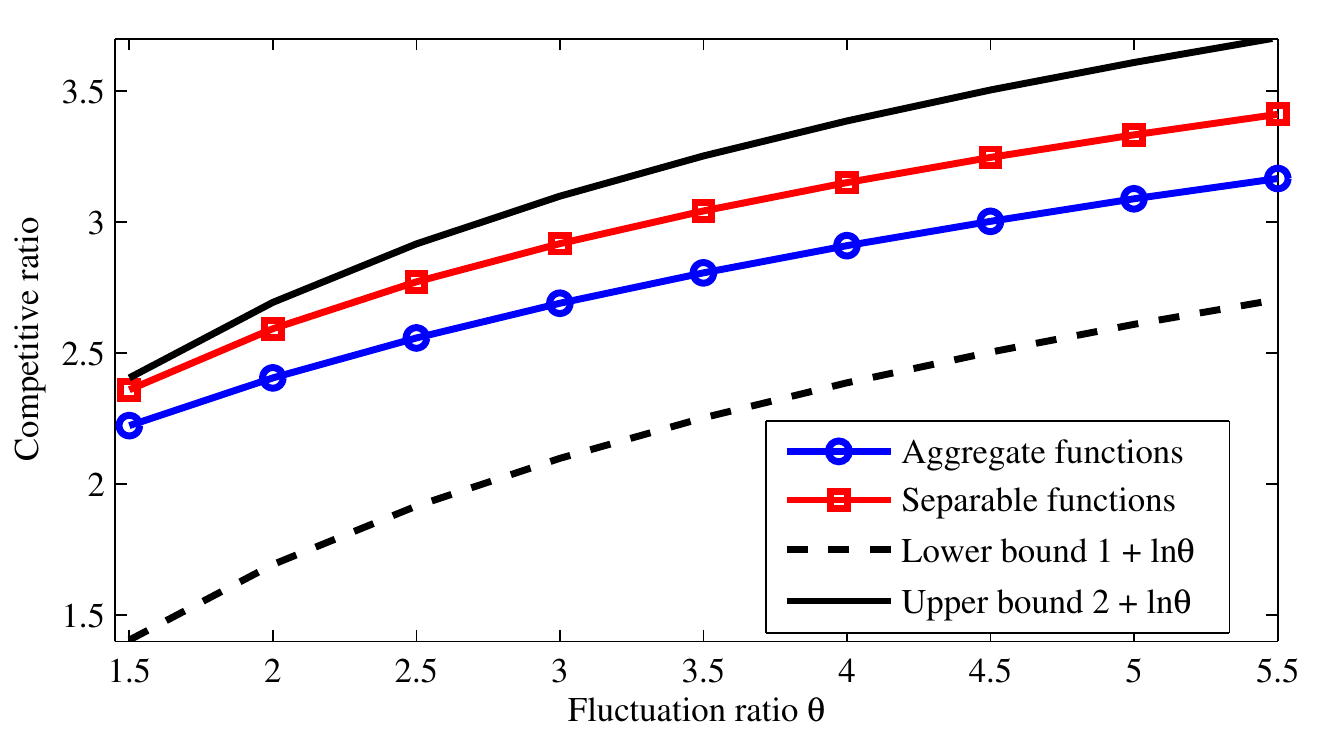}
	\caption{Competitive ratios of $\ota_{\phi}$ for \fomkp with aggregate and separable functions.}
	\label{fig:cr}
\end{figure}
The competitive ratios of both cases are illustrated in Figure~\ref{fig:cr}. In both cases, the competitive ratios are bounded between $1+\ln\theta$ and $2+\ln\theta$, where $1+\ln\theta$ is a lower bound of the optimal competitive ratio. It is also worth contrasting the threshold function with those used in prior work on \omkp.  Compared to~\cite{zhou2008budget}, which uses the same threshold function~\eqref{eq:threshold-got} for all knapsacks, the threshold functions~\eqref{eq:threshold-fomkp-aggregate} and~\eqref{eq:threshold-fomkp-separable} for \fomkp are lower, and consequently estimate a lower marginal cost at the same utilization level, encouraging a more aggressive assignment of items.
The difference in threshold functions results from the rate-limiting constraints and the knapsack-dependent value functions (e.g., separable value functions) in \fomkp. For an example, the rate limits constrain the decision space of \ota, and equivalently increase the chance of missing opportunities for assigning items in the worst case, leading to more aggressive assignments compared to \omkp. Details about impact of rate limits on worst-case instances are provided in Section~\ref{sec:omkp}. 

We close this section by discussing the advantages and limitations of applying the instance-dependent OPD to improve upon competitive ratios of related problems.
The instance-dependent OPD can be considered as an extension of the classical OPD~\cite{buchbinder2009design}. The key advantage of this approach is to provide a pathway to incorporate the understanding of worst-case instances into the design of online algorithms.
Better still, applying this approach does not require a full characterization of the worst case. As is shown in the analysis of \fomkp, a partial understanding of the worst-case instance can already improve upon the competitive ratios.
However, if we have no ideas about the worst-case instances, our approach is no different than the classical OPD. 
For example, with prior knowledge on the worst-case of \fomkp, our approach can be applied to the online $0/1$ knapsack~\cite{zhou2008budget,Zhang2017}, \omkp~\cite{zhou2008budget}, and \omkp with assignment restrictions (\omkpar), for their fractional versions or under an infinitesimal assumption. 
In contrast, for another variant of knapsack problems, including the online multi-dimensional knapsack~\cite{Zhang2017} and online packing problem~\cite{buchbinder2009online}, the worst-case instance of this variant is more difficult to characterize than that of \fomkp due to the strong coupling of the assigned item across multiple knapsacks.
Thus, our approach has no advantage over the classical OPD approach (e.g., in \cite{buchbinder2009online}), and cannot improve upon the known results in~\cite{buchbinder2009online,Zhang2017} at this moment due to a lack of an understanding of worst-case instances.

\section{Optimal Online Algorithms for Generalized One-Way Trading}
\label{sec:got}

In the next two sections we present the analysis that leads to the main results discussed in the previous section. We begin by focusing on an important special case of \fomkp, the generalized one-way trading problem (\got). This problem has garnered considerable interest, e.g.,~\cite{El2001,Lin2019,yang2019online,fujiwara2011average,chin2015competitive,schroeder2018optimal,wang2016competitive}, and serves as a way to introduce the key ideas of our approach without the additional complexity of the full \fomkp formulation.  Then, in Section \ref{sec:omkp} we show how to generalize the ideas presented here to the full \fomkp formulation.  

The key novelty of the main result in this section (Theorem~\ref{thm:threshold-got}) lies in our approach to derive the threshold function~\eqref{eq:threshold-got}, which we outline in Section~\ref{subsec:approach}.  Then, we provide a new proof of the lower bound in Section~\ref{subsec:proof-opt-CR}.  Finally, we discuss extensions to variants of one-way trading in Section~\ref{subsec:variants}

\subsection{Proof of Theorem~\ref{thm:threshold-got}: Designing the Threshold Function}
\label{subsec:proof-threshold-got}

In the \got problem, formulated in~\eqref{p:got}, an operator maintains one knapsack with a total capacity $C$.
Upon the arrival of a new item $n\in\caln$, $\ota_\phi$ immediately decides the fraction of the item to be accepted, $y_n^*$, and obtains a value $g_n(y_n^*)$. 
In this special case of \fomkp, the core pseudo-utility maximization problem~\eqref{p:utility-maximization} in $\ota_\phi$ reduces to the following problem
\begin{align}\label{p:utility-maximization-got}
	\max_{0\le y_n\le D_n}\quad g_n(y_n) - \smallint_{w^{(n)}}^{w^{(n)}+y_n}\phi(u)du.
\end{align}

Our approach here relies on the sufficient conditions on $\phi$ in Lemma~\ref{lem:sufficient}, so we first prove the lemma and then continue with the proof of the theorem. 

\begin{proof}[Proof of Lemma~\ref{lem:sufficient}]
	The dual problem of the offline \got~\eqref{p:got} can be stated as
	\begin{align}\label{p:got-dual}
		\min_{\la \ge 0} \quad&  \sum\nolimits_{n\in\caln} h_n(\la) + \la C,
	\end{align}
	where $\la$ is the dual variable associated with the capacity constraint and 
	\begin{align}\label{eq:conjugate-function}
		h_n(\la) = \max_{0\le y_n \le D_n} g_n(y_n) - \la y_n
	\end{align}
	is the conjugate function of $g_n(\cdot)$. Note that $h_n(\la)$ can be interpreted as the maximal pseudo-utility when a linear price $\la$ is used to estimate the cost of using knapsack capacity. Thus, the maximization problem~\eqref{eq:conjugate-function} has a similar physical meaning to the pseudo-utility maximization problem~\eqref{p:utility-maximization-got} in $\ota_{\phi}$.  This connection is formalized in the following proposition.  
	
	\begin{pro}\label{pro:conjugate-function}
		The conjugate function $h_n(\la)$ has the following properties:
		
		(i) $h_n(\la)$ is a non-increasing function;
		
		(ii) when $\phi(C)\ge U$, $h_n(\phi(w^{(n+1)})) = g_n(y_n^*) - \phi(w^{(n+1)}) y_n^*, \forall n\in\caln$, where $w^{(n+1)} = w^{(n)} + y_n^*$ and $y_n^*$ is the optimal solution to the problem~\eqref{p:utility-maximization-got}.
	\end{pro}
	
	The proof of Proposition~\ref{pro:conjugate-function} is shown in Appendix~\ref{app:proof-conjugate-function}. Property (ii) in Proposition~\ref{pro:conjugate-function} implies that when the linear price $\la$ is set to the marginal cost $\phi(w^{(n+1)})$, the online solution $y_n^*$ of the problem~\eqref{p:utility-maximization-got} also maximizes the problem~\eqref{eq:conjugate-function} in the conjugate function. This relationship connects the online solution and the dual objective, and is important in the \opd analysis. 
	
	Let $w^{(N+1)} := w^{(N+1)}(\cali)$ denote the final utilization of the knapsack after executing the instance $\cali$ by $\ota_\phi$. 
	We divide the set $\Omega$ of all instances into two families $\Omega^1:=\{\cali: 0 \le w^{(N+1)}< \beta \}$ and $\Omega^2:=\{\cali: \beta \le w^{(N+1)}\le C \}$, which contain the instances whose final utilizations fall into the flat segment and the non-decreasing segment, respectively.
	$\Omega^1$ and $\Omega^2$ represent two different types of worst-case instances for \got. 
	$\Omega^1$ contains under-demand instances, in which the knapsack capacity is not used up even when all items are accepted to their weights. Thus, the offline solution is to accept all items.
	$\Omega^2$ includes over-demand instances, in which the capacity can be fully occupied by the offline solution in the worst case.
	This leads to different offline formulations for \opd analysis.

	{\bf Case I: $\cali\in\Omega^1$.} 
	The threshold function $\phi$ estimates the marginal cost of using the knapsack as $L$, the lower bound of the marginal value.
	Thus, all items in $\cali$ are accepted to their weights by $\ota_\phi$ and we have $\alg(\cali,\phi) = \sum_{n\in\caln} g_n(D_n)$.
	We can build an upper bound $\dual(\cali,\phi)$ of the offline optimum $\opt(\cali)$ by constructing a feasible dual solution $\hat{\la}$. 
	A natural choice of the feasible dual solution is $\hat{\la} = \phi(w^{(N+1)}) = L$, which is the marginal cost of the knapsack for packing one more unit of item. Substituting this dual solution to the dual objective in~\eqref{p:got-dual} gives $\dual(\cali,\phi) = \sum\nolimits_{n\in\caln} h_n(L) + L C = \sum\nolimits_{n\in\caln} g_n(D_n) + L (C - w^{(N+1)})$. However, $\dual(\cali,\phi)$ cannot be further bounded by $\alpha \alg(\cali,\phi)$, which can be observed when $w^{(N+1)}\to 0$.
	This is because the capacity parameter $C$ in the dual objective~\eqref{p:got-dual} is not appropriate for the under-demand instances whose capacity constraint will not be binding in the offline problem. 
	
	Instead of using $C$ as the capacity parameter in the offline problem~\eqref{p:got}, we can change it by adding $\sum_{n\in\caln}y_{n} \le w^{(N+1)}$ to the offline formulation. This change will not affect the offline solution for $\cali \in \Omega^1$ since the total accepted demand by offline problem cannot exceed the total weights of all items. In this way, the dual objective is changed to $\sum\nolimits_{n\in\caln} h_n(\hat\la) + \hat\la w^{(N+1)}$ and we have
	\begin{align}
		\opt(\cali) \le \sum\nolimits_{n\in\caln} h_n(L) + L w^{(N+1)} = \sum\nolimits_{n\in\caln} g_n(D_n)  = \alg(\cali,\phi).
	\end{align}
	Thus, we have $\opt(\cali)/\alg(\cali,\phi) \le 1, \forall \cali\in\Omega^1$.

	{\bf Case II: $\cali\in\Omega^2$.} An adversary can always add one more item with weight $C$ and marginal value $\phi(w^{(N+1)})$. This new item will be rejected by $\ota_{\phi}$ while the offline optimum will accept this item making the knapsack fully occupied.
	In this case, we can keep using the dual objective~\eqref{p:got-dual} and set the feasible dual solution to $\hat{\la} = \phi(w^{(N+1)})$.
	Based on weak duality, we have
	\begin{subequations}
		\label{eq:proof-got-opd}
		\begin{align}
			\nonumber
			\opt(\cali) &\le \sum\nolimits_{n\in\caln} h_n(\phi(w^{(N+1)})) + \phi(w^{(N+1)}) C\\
			\label{eq:proof-got-property1}
			&\le \sum\nolimits_{n\in\caln} h_n(\phi(w^{(n+1)})) + \phi(w^{(N+1)}) C\\
			\label{eq:proof-got-property2}
			&= \sum\nolimits_{n\in\caln} [g_n(y_n^*) - \phi(w^{(n+1)})y_n^*] + \phi(w^{(N+1)}) C\\
			\label{eq:proof-got-ineq1}
			&\le \sum\nolimits_{n\in\caln}g_n(y_n^*) + \phi(w^{(N+1)})C - \smallint_{0}^{w^{(N+1)}}\phi(u)du\\
			\label{eq:proof-got-sufficient-conds}
			&\le \sum\nolimits_{n\in\caln}g_n(y_n^*) + (\alpha -1)\smallint_{0}^{w^{(N+1)}}\phi(u)du\\
			\label{eq:proof-got-utility-max}
			&\le \sum\nolimits_{n\in\caln}g_n(y_n^*) + (\alpha -1)\sum\nolimits_{n\in\caln}g_n(y_n^*) = \alpha \alg(\cali).
		\end{align}
	\end{subequations}
	Based on the properties of the conjugate function in Proposition~\ref{pro:conjugate-function}, we can have equations~\eqref{eq:proof-got-property1} and \eqref{eq:proof-got-property2} when $\phi(C) \ge U$. Since $\phi$ is a non-decreasing function, we have $\phi(w^{(n+1)})y_n^* \ge \smallint_{w^{(n)}}^{w^{(n+1)}}\phi(u)du$ and $\sum\nolimits_{n\in\caln}\phi(w^{(n+1)})y_n^* \ge \sum\nolimits_{n\in\caln}\smallint_{w^{(n)}}^{w^{(n+1)}}\phi(u)du = \smallint_{0}^{w^{(N+1)}}\phi(u)du$. Inequality~\eqref{eq:proof-got-ineq1} holds. If the threshold function $\phi$ satisfies the differential equation~\eqref{eq:ode-sufficient} in Lemma~\ref{lem:sufficient}, we can have inequality~\eqref{eq:proof-got-sufficient-conds}.  
	Based on the pseudo-utility maximization problem~\eqref{p:utility-maximization-got}, non-negative utility is achieved for each $n\in\caln$, i.e., $g_n(y_n^*) \ge \smallint_{w^{(n)}}^{w^{(n+1)}}\phi(u)du, \forall n\in\caln$. Thus, we have $\sum_{n\in\caln}g_n(y_n^*)\ge\smallint_{0}^{w^{(N+1)}}\phi(u)du$, and this gives inequality~\eqref{eq:proof-got-utility-max}.
	Thus, we have $\opt(\cali)/ \alg(\cali) \le \alpha, \forall \cali \in\Omega^2$ if the sufficient conditions in Lemma~\ref{lem:sufficient} are satisfied.

	Finally, combining the two cases completes the proof.
\end{proof}

Continuing with the proof of Theorem \ref{thm:threshold-got}, we next prove that $\phi^*$ in~\eqref{eq:threshold-got} achieves the smallest competitive ratio among all threshold functions that satisfy the sufficient conditions in Lemma~\ref{lem:sufficient}. To do so, we make use of Gronwall's Inequality, summarized below.

\begin{lem}[Gronwall's Inequality, Theorem 1, p.356, \cite{Mitrinovic1991}, and Lemma 4, \cite{Jones1964}]\label{lem:Gronwall-inequality}
	Let $f(x)$ be a function defined on $[\underline{x},\overline{x}]$ either continuous or of bounded variation. Let $a(x)$ and $b(x)$ be integrable functions, and $b(x)\ge 0$ for $x\in [\underline{x},\overline{x}]$. We can claim the following statements.
	
	(i) If $f(x) \ge a(x) + b(x)\smallint_{\underline{x}}^{x} f(u) du, x\in[\underline{x},\overline{x}]$, then we have
	\begin{align}\label{eq:Gronwall-result1}
		f(x) \ge a(x) + b(x)\smallint_{\underline{x}}^{x} a(u) \exp({\smallint_{u}^{x}b(s)ds}) du,\quad x\in[\underline{x},\bar{x}].
	\end{align}
	
	(ii) The result remains valid if $\ge$ is replaced by $\le$ in both conditions and results of statement (i).
	
	(iii) Equation \eqref{eq:Gronwall-result1} holds in equality for $x\in[\underline{x},\overline{x}]$ if the condition holds in equality for $x\in[\underline{x},\overline{x}]$.
\end{lem}

Applying Gronwall's Inequality to the differential equation in~\eqref{eq:ode-sufficient} gives
\begin{align}
	\varphi(w) \le \frac{\alpha L \beta}{C} + \frac{\alpha}{C}{\smallint_{\beta}^{w}}\frac{\alpha L \beta}{C} e^{(w-u)\alpha/C} du = \frac{\alpha L \beta}{C} e^{(w-\beta)\alpha/C}, w\in[\beta,C].
\end{align}
Since $\varphi(C) \ge U$, we have
$U \le \varphi(C) \le \frac{\alpha L \beta}{C} e^{(C-\beta)\alpha/C}$.
Combining with the sufficient condition~\eqref{eq:ode-sufficient}, we can conclude that the smallest $\alpha$ is achieved when all above inequalities hold in equality. It is equivalent that all inequalities in the sufficient condition hold in equality based on statement (iii) in Lemma~\ref{lem:Gronwall-inequality}. Solving those equality equations gives the threshold function $\phi^*$ in~\eqref{eq:threshold-got} and the resulting competitive ratio is $\alpha_{\phi^*} = 1 + \ln\theta$.

\subsection{Proof of Theorem~\ref{thm:opt-CR}: Bounding the Optimal Competitive Ratio}
\label{subsec:proof-opt-CR}

Note that the optimal competitive ratio achievable for the one-way trading problem has been shown to be $1+\ln\theta$~\cite{yang2019online}.
Since one-way trading is a special case of \got, the competitive ratio of \got is also lower bounded by $1+\ln\theta$.
Thus, Theorem~\ref{thm:threshold-got} equivalently shows that $\ota_{\phi^*}$ achieves the optimal competitive ratio. 
However, our goal in this section is to provide a new proof of the optimal competitive ratio based on understanding special instances. This, in turn, builds a connection between the online algorithm and the worst-cast instance. 

Our approach is to first characterize a necessary condition that any $\alpha$-competitive online algorithm must satisfy, and then derive the lower bound as the minimal $\alpha$ ensuring that there exist online algorithms satisfying the necessary condition.  The necessary condition is constructed based on a subset of instances $\Omega_{CN} \subseteq \Omega$ called continuously non-decreasing instances.

\begin{dfn}\label{dfn:instance}
	An instance is called $p$-continuously non-decreasing, $p\in[L, U]$, if 
	\begin{itemize}
		\item the instance is composed of a sequence of items indexed by $n\in\caln$. Each item has a linear value function $g_n(y_n) = v_n y_n$ and its weight is $D_n = C$. 
		\item the marginal value of the first item is $L$, i.e., $v_1 = L$
		\item the increment of the marginal values between successive items is non-negative and arbitrarily small, i.e., $0\le v_{n+1} - v_{n} \le \epsilon$, where $\forall \epsilon > 0$. 
		\item the marginal value of the last item is $p$, i.e., $v_N = p$.
	\end{itemize}
\end{dfn}
Let $\cali_p$ denote the $p$-continuously non-decreasing instance and let $\Omega_{CN} := \{\cali_p\}_{p\in[L,U]}$.

\begin{dfn}[Utilization Function]
	A utilization function $\psi(p):[L,U]\to[0,C]$ is defined as the final utilization of the knapsack after executing the instance $\cali_p$ by an online algorithm.
\end{dfn}
Note that every online algorithm can be mapped to a utilization function via $\Omega_{CN}$.
The key to our approach here is Lemma~\ref{lem:necessary-condition} and we next show that the utilization function of any $\alpha$-competitive online algorithm must satisfy the necessary condition in Lemma~\ref{lem:necessary-condition}.

\begin{proof}[Proof of Lemma~\ref{lem:necessary-condition}]
	Since online algorithms make real-time irrevocable decisions only based on causal information, $\psi(p)$ is a non-decreasing function in $[L,U]$. Since the maximum utilization is $C$, the utilization function must satisfy the boundary condition $\psi(U) \le C$. Additionally, by definition, the total value achieved by an $\alpha$-competitive online algorithm is at least $1/\alpha$ of the offline optimum for any arrival instances.    Thus, under the instance $\cali_L$, we have
	\begin{align*}
		{\opt}(\cali_L) &= LC \quad {\text{and}}\quad {\alg}(\cali_L) =L \psi(L),
	\end{align*}
	and an $\alpha$-competitive algorithm must ensure ${\alg}(\cali_L) \ge {\opt(\cali_L)}/{\alpha}$, which gives the boundary condition $\psi(L) \ge {C}/{\alpha}$.
	
	More specifically, under the instance $\cali_p, p\in(L,U]$, we have
	\begin{align*}
		{\opt}(\cali_p) = p C, \quad {\text{and}}\quad {\alg}(\cali_p) = L\psi(L)+\smallint_{L}^{p} u d\psi(u),
	\end{align*}
	where $u d\psi(u)$ denotes the value achieved by the item with marginal value $u$. An $\alpha$-competitive algorithm must ensure ${\alg}(\cali_p) \ge {{\opt}(\cali_p)}/{\alpha}$ which gives the differential equation in~\eqref{eq:necessary-condition}. Combining all above conditions gives the necessary condition \eqref{eq:necessary-condition}.
\end{proof}

Finally, to complete the proof of the theorem we  derive the minimal $\alpha$ that can ensure there exists a non-decreasing utilization function $\psi$ satisfying the necessary condition~\eqref{eq:necessary-condition}.
Using integration by parts, we have
$L\psi(L)+{\smallint_{L}^{p}} u d\psi(u) = L\psi(L)+[\psi(u)u ]|_{L}^p - {\smallint_{L}^{p}} \psi(u)du = \psi(p)p - {\smallint_{L}^{p}} \psi(u)du$.
Combining above equation and the necessary condition \eqref{eq:necessary-condition}, we see that the utilization function $\psi$ corresponding to any $\alpha$-competitive online algorithms must satisfy 
$\psi(p)p - \smallint_{L}^{p} \psi(u)du \ge pC/\alpha, p\in[L,U]$.
Applying Gronwall's Inequality in Lemma \ref{lem:Gronwall-inequality}, we obtain
\begin{align*}
	\psi(p) \ge \frac{C}{\alpha} + \frac{1}{p} {\smallint_{L}^{p}} \frac{C}{\alpha} \exp({\smallint_{u}^{p}\frac{1}{s}ds}) du = \frac{C}{\alpha} \cdot \left[1 + {\smallint_{L}^{p}} \frac{1}{u}du \right]
	= \frac{C}{\alpha} \cdot \left[1 + \ln\left(\frac{p}{L}\right) \right]. 
\end{align*}
Since $\psi(U) \le C$, we have
$\frac{C}{\alpha} \left[1 + \ln\theta \right] = \psi(U) \le C$, which gives $\alpha \ge 1 + \ln\theta$.
And the minimal $\alpha = 1 + \ln\theta$ can be achieved when inequalities in~\eqref{eq:necessary-condition} all hold in equality. Thus, $1+\ln\theta$ is a lower bound of the competitive ratio.

\subsection{Two Variants of \got} \label{subsec:variants}

In order to show the generality of our approach for \got, we further devise $\ota_\phi$ for two variants of \got using the approach.  In both cases we obtain results that match or improve the state-of-the-art.

\subsubsection*{Variant 1: \got without leftover capacity}
This variant considers the classical setting of the one-way trading problem in which after the last item, the remaining capacity of the knapsack, if any, can be used to pack items with the lowest marginal value $L$. 
When the value function is linear, this variant is studied by \cite{El2001}. It is solved using a threat-based online algorithm, and the optimal competitive ratio that is achieved is the solution of the equation $\alpha = \ln\frac{U-L}{\alpha L - L}$. Concretely, the offline formulation of this variant can be stated as   
\begin{align}\label{p:varaint1-got}
	\max_{y_n} \quad \sum\nolimits_{n\in\caln} g_n(y_n) + \left(C - \sum\nolimits_{n\in\caln} y_n \right)L \quad {\rm s.t.} \sum\nolimits_{n\in\caln}y_n \le C, \quad 0 \le y_n \le D_n, \forall n\in\caln.
\end{align}
Note that this variant cannot be considered as a \got with a value function $g_n(y_n) - L y_n$ since (i) its marginal value is lower bounded by $0$, which does not satisfy Assumption~\ref{ass:value-function}, and (ii) the total value is lower bounded by $CL$ even when no item is accepted by an online algorithm. 
The following Corollary~\ref{lem:threshold-varaint1-got} (see proof in Appendix~\ref{app:proof-variant1-got-threshold}) shows that we can design $\ota_\phi$ to achieve the optimal competitive ratio of this variant.

\begin{cor}\label{lem:threshold-varaint1-got}
	Under Assumption~\ref{ass:value-function}, if the threshold function of $\ota_\phi$ for Variant 1 of \got is
	\begin{align}\label{eq:threshold-got-variant1}
		\phi^*(w)= L + (U-L) e^{\frac{\alpha_{\phi^*}}{C}w - \alpha_{\phi^*}}, w\in[0, C],
	\end{align} 
	the competitive ratio $\alpha_{\phi^*}$ of $\ota_{\phi^*}$ is the solution of the equation $\alpha_{\phi^*} = \ln\frac{U-L}{\alpha_{\phi^*} L - L}$.
\end{cor}

\begin{cor}\label{lem:varaint1-got}
	The optimal competitive ratio for Variant 1 of \got is the solution of $\alpha^* =  \ln\frac{U-L}{\alpha^* L - L}$.
\end{cor}

We can prove Corollary~\ref{lem:varaint1-got} using the same approach in Section~\ref{subsec:proof-opt-CR} and the worst-case instance is still the continuously non-decreasing instance $\Omega_{CN}$. The detail is presented in Appendix~\ref{app:proof-varaint1-got-opt}.

\subsubsection*{Variant 2: Relaxed \got}
In this variant, condition (iii) in Assumption~\ref{ass:value-function} is relaxed to the following: 
\begin{ass}\label{ass:relaxed}
	The derivative of the value function satisfies $ L \le g_n'(0) \le U$ and $ {g_n(D_n)}/{D_n} \ge {L}/{c}, \forall n\in\caln$, where $c \ge 1$ is a given parameter.
\end{ass}
Assumption~\ref{ass:relaxed} bounds the marginal value of the value function at origin between $L$ and $U$, and the average value is lower bounded by $L/c$. This new assumption allows a broader class of value functions whose marginal values may reach $0$ (e.g., quadratic functions that can reach their maximums).
The assumption has also been introduced by~\cite{Lin2019}, in which a CR-Pursuit online algorithm is proposed to solve \got (without rate limits) and is shown to achieve a competitive ratio upper bounded by $O(c(\ln(\theta)+1))$.  Our approach yields the following result.

\begin{cor}\label{lem:varaint2-got}
	Under the conditions (i) and (ii) in Assumption~\ref{ass:value-function} and Assumption~\ref{ass:relaxed}, when the threshold function of $\ota_\phi$ for Variant 2 of \got is given by
	\begin{align}\label{eq:variant2-got-threshold-fun}
		\phi^*(w) = 
		\begin{cases}
			\frac{L}{c} \cdot \frac{e^{w/C} - 1}{e^{\beta^*/C} - 1} & w\in[0,\beta^*)\\
			\frac{L}{c} \cdot e^{(w-\beta^*)\ln(c\theta)/(C-\beta^*)} & w \in [\beta^*, C]
		\end{cases},
	\end{align}
	where $\beta^* = (W(c\theta\ln(c\theta)/e) - \ln(c\theta) + 1)C$, the competitive ratio of $\ota_{\phi^*}$ is
	$ \frac{\ln(c\theta)}{\ln(c\theta) - W(c\theta \ln(c\theta)/e) }$.
\end{cor}

In the corollary above, $W(\cdot)$ is the Lambert-$W$ function, which is the inverse function of $f(x) = xe^x$. 
Since $W(x) \le \ln(x) - \ln\ln(x) + O(1)$, we have $\ln(c\theta) - W(c\theta \ln(c\theta)/e) \ge O(1)$. Consequently, the competitive ratio achieved by $\ota_{\phi^*}$ in Corollary~\ref{lem:varaint2-got} is $O(\ln(c\theta))$, improving the upper bound in~\cite{Lin2019} from linear order $O(c(\ln(\theta)+1))$ in $c$ to logarithmic order.

\section{Competitive Algorithms for \fomkp}
\label{sec:omkp}

In this section, we prove our main results, which bound the competitive ratio for the general form of \fomkp. To do this, we use the same general approach as illustrated in the previous section for \got. However, the generality of \fomkp adds considerable complexity to this case. 
We primarily focus on the proof of Theorem~\ref{thm:threshold-omkp-aggregate} for the \fomkp with aggregate value functions. The proof of Theorem~\ref{thm:threshold-omkp-separable} for the separable functions proceeds much the same. Thus, we highlight the key differences here and defer the full proof to Appendix~\ref{app:proof-thm-fomkp-separable}.

\textbf{Proof of Theorem~\ref{thm:threshold-omkp-aggregate}: Aggregate Functions.} First, we construct a counterpart to Lemma~\ref{lem:sufficient} for \got, providing sufficient conditions for designing the threshold function. The sufficient condition on the threshold function of each knapsack is not a trivial extension of the single knapsack case.

\begin{lem}
	\label{lem:fomkp-sufficient-aggregate}
	Under Assumption~\ref{ass:value-function}, $\ota_\phi$ for \fomkp with aggregate value functions is $\alpha$-competitive if the threshold function $\phi=\{\phi_m\}_{m\in\calm}$ is in the form of, $\forall m\in\calm$, 
	\begin{align*}
		\phi_m(w)=
		\begin{cases}
			L&  w\in[0,\beta_m)\\
			\varphi_m(w) & w\in[\beta_m, C_m]
		\end{cases},
	\end{align*}
	where $\beta_m \in [0,C_m]$ is a utilization threshold and $\varphi_m$ is a non-decreasing function, and $\phi_m$ satisfies
	\begin{align}\label{eq:fomkp-aggregate-sufficient}
		\begin{cases}
			\varphi_m(w)C_m \le \alpha \smallint_{0}^{w}\phi_m(u)du - L\beta_m, \quad  w\in[\beta_m,C_m],\\
			\varphi_m(\beta_m) = L,\varphi_m(C_m) \ge U.
		\end{cases}
	\end{align} 
\end{lem}

To prove Lemma~\ref{lem:fomkp-sufficient-aggregate}, we divide the set of instances $\Omega$ into three subsets $\Omega^1$, $\Omega^2$, and $\Omega^3$. The instances in those subsets result in different worst cases. Thus, we construct instance-dependent dual objectives to bound the offline optimum in each case, leading to the sufficient conditions in Lemma~\ref{lem:fomkp-sufficient-aggregate}.  
We sketch the proof of Lemma~\ref{lem:fomkp-sufficient-aggregate} here and include the full version in Appendix~\ref{app:proof-fomkp-aggregate}.

\begin{proof}[Proof Sketch of Lemma \ref{lem:fomkp-sufficient-aggregate}]
	Let $w_m^{(N+1)}:= w_m^{(N+1)}(\cali)$ denote the final utilization of the knapsack $m$ after executing instance $\cali$ by $\ota_\phi$. 
	$\Omega^1 := \{\cali: 0 \le w_m^{(N+1)} < \beta_m, \forall m\in\calm \}$ and $\Omega^2 := \{\cali: \beta_m\le w_m^{(N+1)} \le C_m, \forall m\in\calm \}$ contain the instances whose final utilizations of all knapsacks are below and above their utilization thresholds $\beta_m$, respectively. Excluding these two subsets, the remaining instances form $\Omega^3 := \Omega\setminus(\Omega^1\cup \Omega^2)$, in which some knapsacks $\calm^1:=\{m\in\calm: 0 \le w_m^{(N+1)} < \beta_m \}$ have final utilizations below the utilization thresholds and the others $\calm^2:=\{m\in\calm: \beta_m \le w_m^{(N+1)} < C_m \}$ have final utilizations above the thresholds.
	
	The cases when $\cali\in\Omega^1$ and $\cali\in\Omega^2$ correspond to general versions of Case I and Case II in the proof of Lemma~\ref{lem:sufficient} for \got.
	The idea is to use the optimal primal and dual variables of the pseudo-utility maximization problem to construct the feasible dual solution in the \opd analysis, and decompose the dual objective into a summation of equations corresponding to individual knapsacks, making these cases similar to those in \got.  
	
	The main new challenge comes from Case III, in which the knapsacks are coupled in a non-trivial way.
	The key difference between Case II and Case III is that the knapsacks in $\calm^1$ may not be fully occupied by the offline solution under the worst-case instance in $\Omega^3$.
	This is because the total amount of items, which can be packed into $\calm^1$, is limited by $\sum_{m\in\calm^1}w_{m}^{(N+1)} + \sum_{m\in\calm^2} \beta_m$.
	Compared to the online solution, the additional amount of items that can be assigned to the knapsacks in $\calm^1$ in the offline solution is upper bounded by the total amount of items that is assigned to $\calm^2$ before reaching the utilization threshold. 
	The marginal cost of the assigned items above the utilization threshold in $\calm^2$ is larger than $L$. 
	Therefore, the reason why such items are not assigned to $\calm^1$ is that the items are not feasible for such assignment due to the rate limits. 
	Thus, those items cannot be assigned to $\calm^1$ in the offline solution as well.
	Based on this understanding of the worst-case instance, we add the following constraint to the offline formulation
	\begin{align*}
		\sum\nolimits_{n\in\caln}\sum\nolimits_{m\in\calm^1}y_{nm} \le \sum\nolimits_{m\in\calm^1} w_m^{(N+1)} + \sum\nolimits_{m\in\calm^2} \beta_m.
	\end{align*}
	Applying \opd analysis to the new offline problem gives the sufficient condition in Lemma~\ref{lem:fomkp-sufficient-aggregate}.
\end{proof}

Now, to complete the proof of Theorem~\ref{thm:threshold-omkp-aggregate}, we apply Gronwall's Inequality to the differential equation~\eqref{eq:fomkp-aggregate-sufficient} and obtain
\begin{align*}
	\varphi_m(w) \le \frac{\alpha - 1}{C_m} L \beta_m + \frac{\alpha}{C_m} {\smallint_{\beta_m}^{w}} \frac{\alpha - 1}{C_m}L\beta_m e^{\alpha(w-u)/C_m}du = \frac{\alpha-1}{C_m} L\beta_m e^{\alpha(w - \beta_m)/C_m}, w\in[\beta_m,C_m].
\end{align*}
Since $\varphi_m(C_m) \ge U$, we have $U \le \varphi(C_m) \le \frac{\alpha-1}{C_m} L\beta_m e^{\alpha(C_m - \beta_m)/C_m}$. The minimum $\alpha$ is achieved when all inequalities in the sufficient condition~\eqref{eq:fomkp-aggregate-sufficient} hold in equality. This gives $U = \frac{\alpha-1}{C_m} L\beta_m e^{\alpha(C_m - \beta_m)/C_m}$ and $\beta_m = \frac{C_m}{\alpha - 1}$. Thus, the resulting competitive ratio $\alpha_{\phi^*}$ is the solution of the equation $\alpha_{\phi^*} - 1 - \frac{1}{\alpha_{\phi^*} - 1} = \ln\theta$ and the threshold function is given by \eqref{eq:threshold-fomkp-aggregate}.

\textbf{Proof Sketch of Theorem~\ref{thm:threshold-omkp-separable}: Separable Functions.} Compared to Theorem~\ref{thm:threshold-omkp-aggregate}, the key difference in proving Theorem~\ref{thm:threshold-omkp-separable} occurs in Case III.  Cases I and II proceed similarly in both cases but, for separable value functions, the total amount of items that can be reassigned from knapsacks in $\calm^2$ to knapsacks in $\calm^1$ is upper bounded by $\sum_{m\in\calm^2}w_m^{(N+1)}$ instead of $\sum_{m\in\calm^2} \beta_m$. This is because each knapsack is associated with an independent value function, and thus the marginal utility, which determines the assignment of a small bit of item, depends on both the marginal value of the item and the marginal cost of the knapsack. So, the reason that the items are assigned to knapsacks in $\calm^2$ may not be due to the rate limits restricting the assignment from knapsacks in $\calm^1$.  
Instead, this may happen because assigning to the knapsacks in $\calm^2$ can result in higher marginal utility. 
In this case, we add a new constraint to the offline problem and the resulting dual objective finally leads to a different threshold function and competitive ratio in Theorem~\ref{thm:threshold-omkp-separable}.

\section{Case Study}
\label{sec:exp} 
This section presents a brief demonstration of our proposed algorithm in the context of the EV charging problem. The experiments are not meant to be exhaustive, rather they are intended to validate the theoretical results and illustrate the potential of the approach. We consider a system consisting of multiple stations working parallel where EVs can charge.  The power capacity is limited, and is much smaller than the total power demanded by the vehicles. Therefore, the station cannot admit the total demands of all vehicles and must decide the amount of power to allocate to the new vehicle upon its arrival.

\textbf{Experimental Setup.} We use the Adaptive Charging Network Dataset, ACN-Data, which includes over 50,000 EV charging sessions, and 54 charging stations~\cite{LeeLiLow2019a}. We use a sequence of more than 2,000 charging sessions. The dataset includes information about the arrival and departure time, and the power demand. In our experiments, we consider one-hour time slots within a time horizon of one day, hence, the total number of time slots, $M$, is chosen as 24. 

We compare our online algorithm with a fixed threshold algorithm, \textsf{FTA}, which admits an item if its value is above a fixed threshold of $\sqrt{U\times L}$, and then it delivers the maximum possible supply power up to the vehicle's demand or the station's capacity. While no prior work exists on the problem studied in this paper, \textsf{FTA} is the most common approach for online knapsack problems, and provides a contrast to the utilization-based threshold in our proposed \ota algorithm. The threshold value $\sqrt{U\times L}$ is selected because our focus is on improving the worst-case performance and this value achieves the best possible competitive ratio among the fixed threshold policies~\cite{El2001}. 

We use a linear, and also the quadratic value function, following the assumption of previous work such as~\cite{zheng2014online}. In our experiments, we set the value fluctuation ratio $\theta = 36$. The results are not too sensitive to this choice. We evaluate the performance of our algorithm in three different congestion levels: low, medium, and high, where the system is able to cover roughly 55\%, 10\%, and 2.7\% of the demand, respectively. 
For each instance, we randomly generate 20 trials for each day, each with different values, and report the average results for $90\times20 = 1800$ trials. 
Last, we report the \textit{empirical profit ratio} of different algorithms, which is the ratio between the profit obtained by the offline optimal solution and that of an online algorithm in experiments. 
\begin{figure*}
	\centering
	\subfigure[Linear value function and low congestion]{\includegraphics[width=0.24\textwidth]{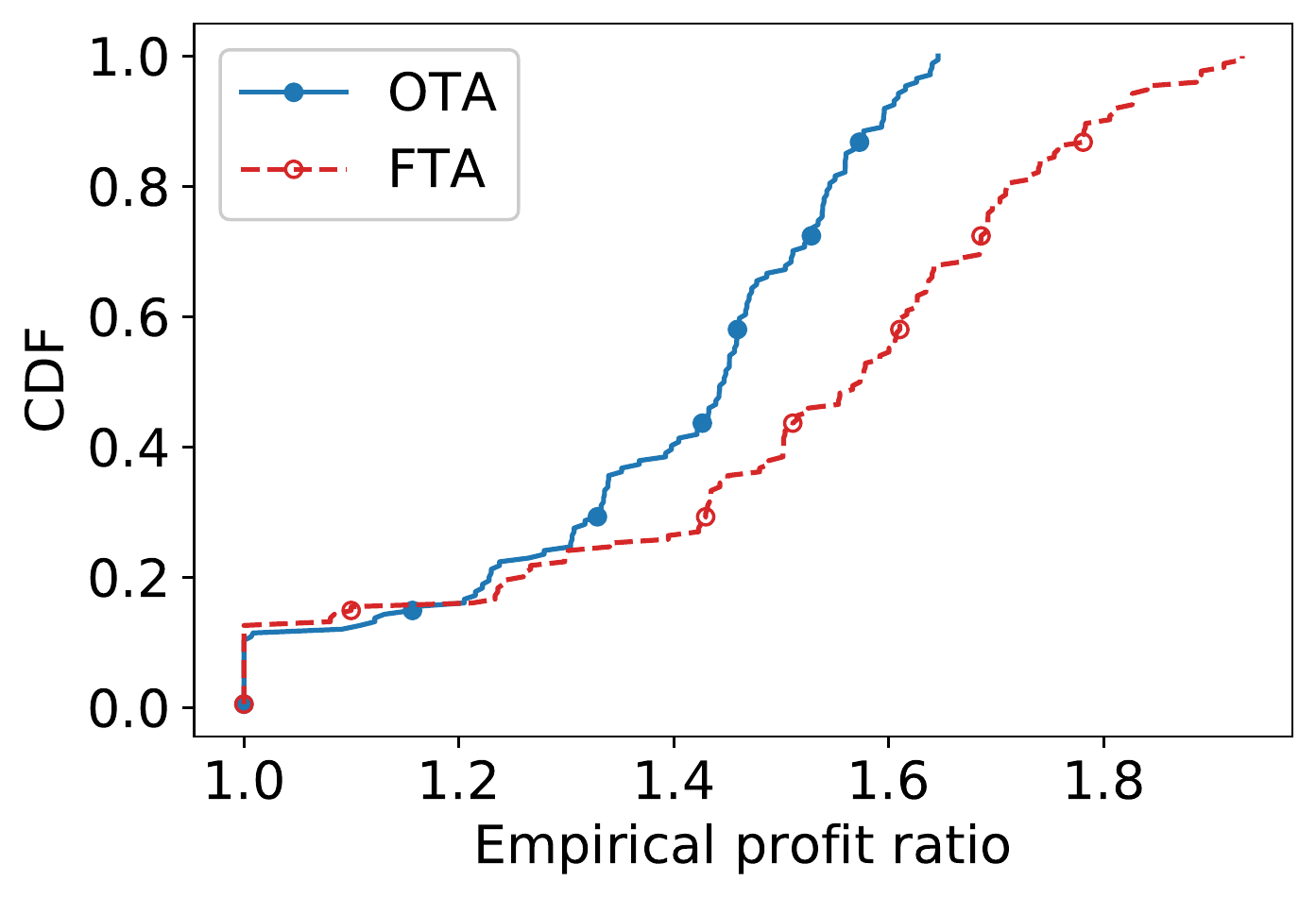}}
	\subfigure[Linear value function and medium congestion]{\includegraphics[width=0.24\textwidth]{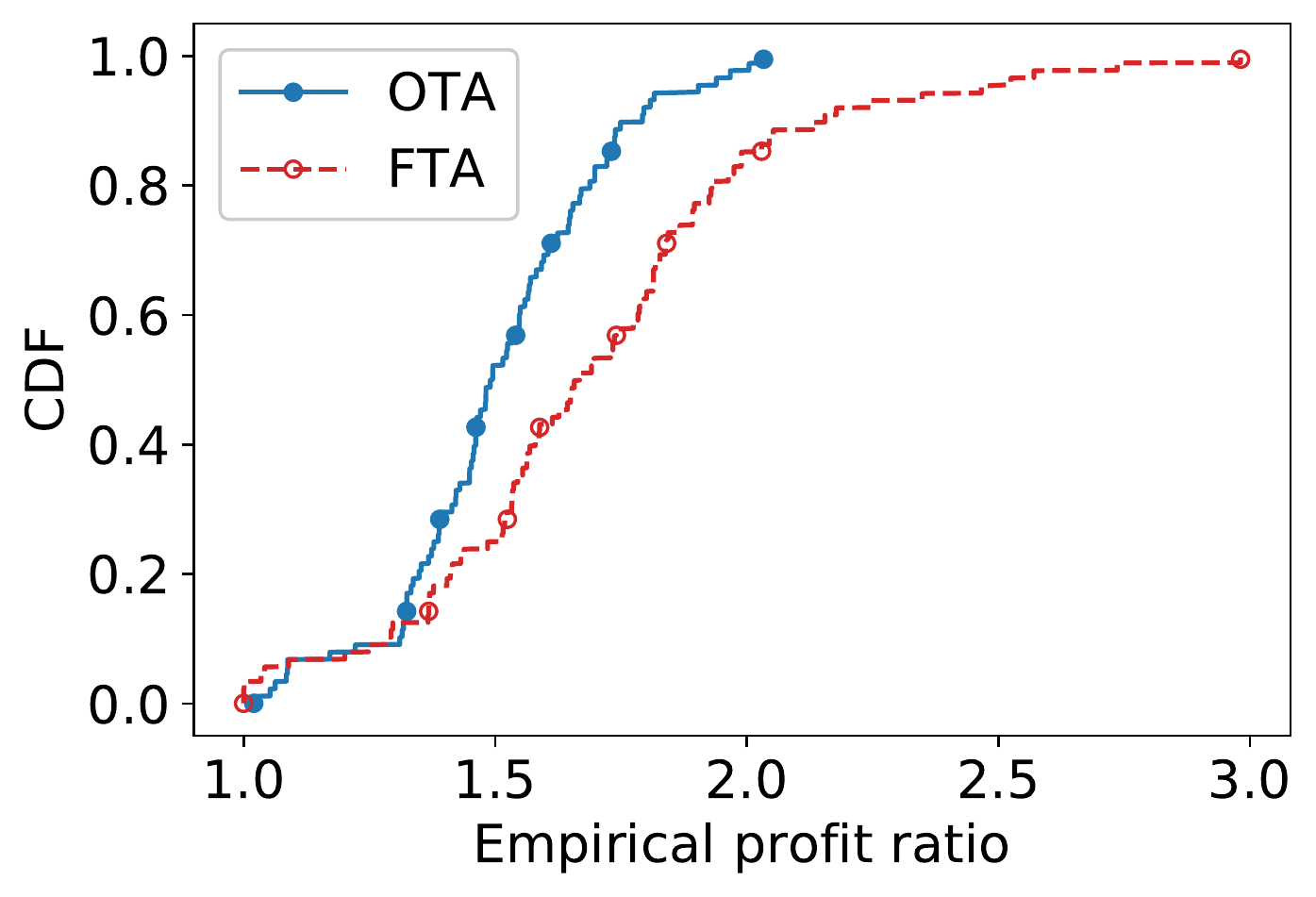}}
	\subfigure[Linear value function and high congestion]{\includegraphics[width=0.24\textwidth]{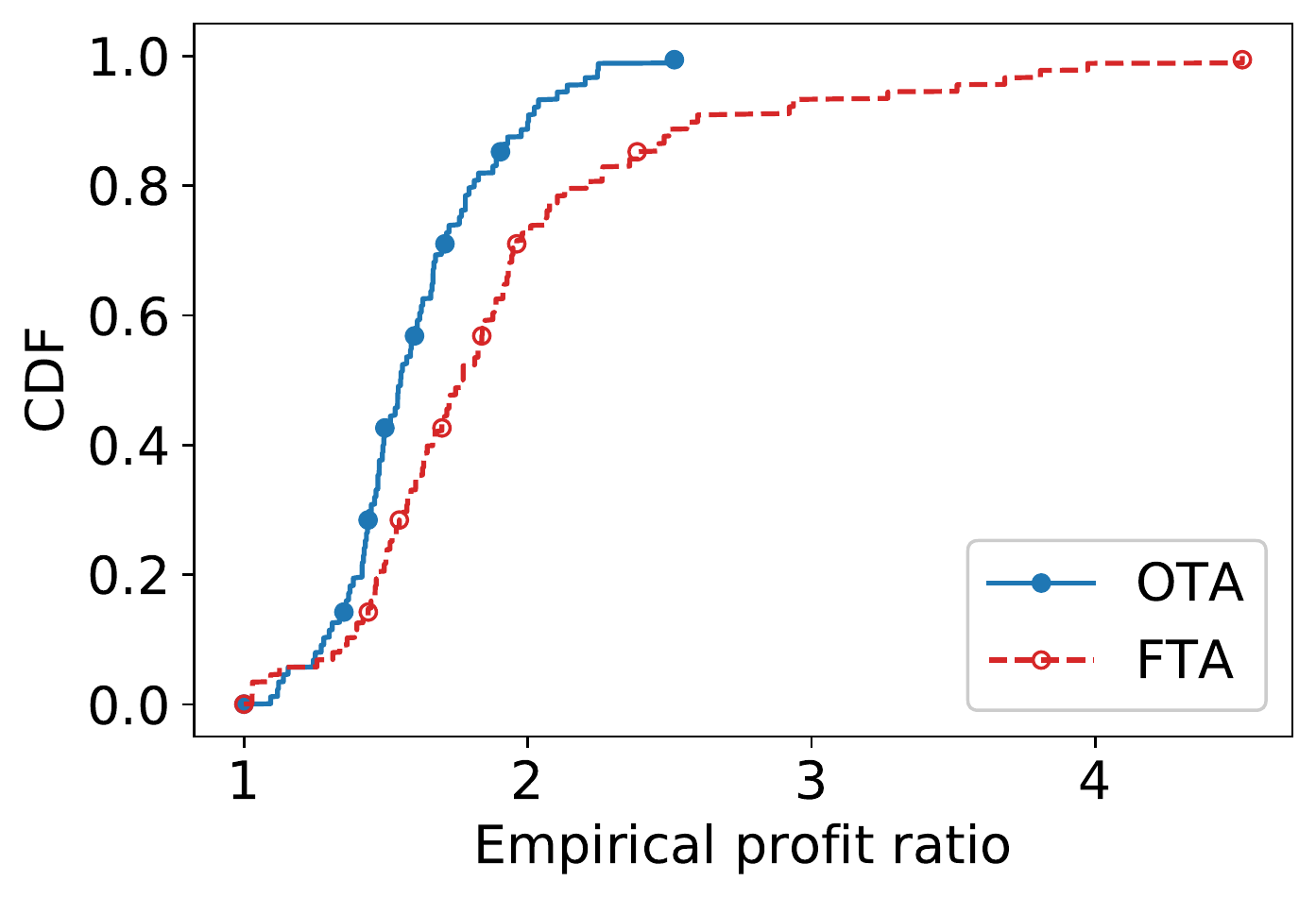}}
	\subfigure[Linear value function and varying congestion]{\includegraphics[width=0.24\textwidth]{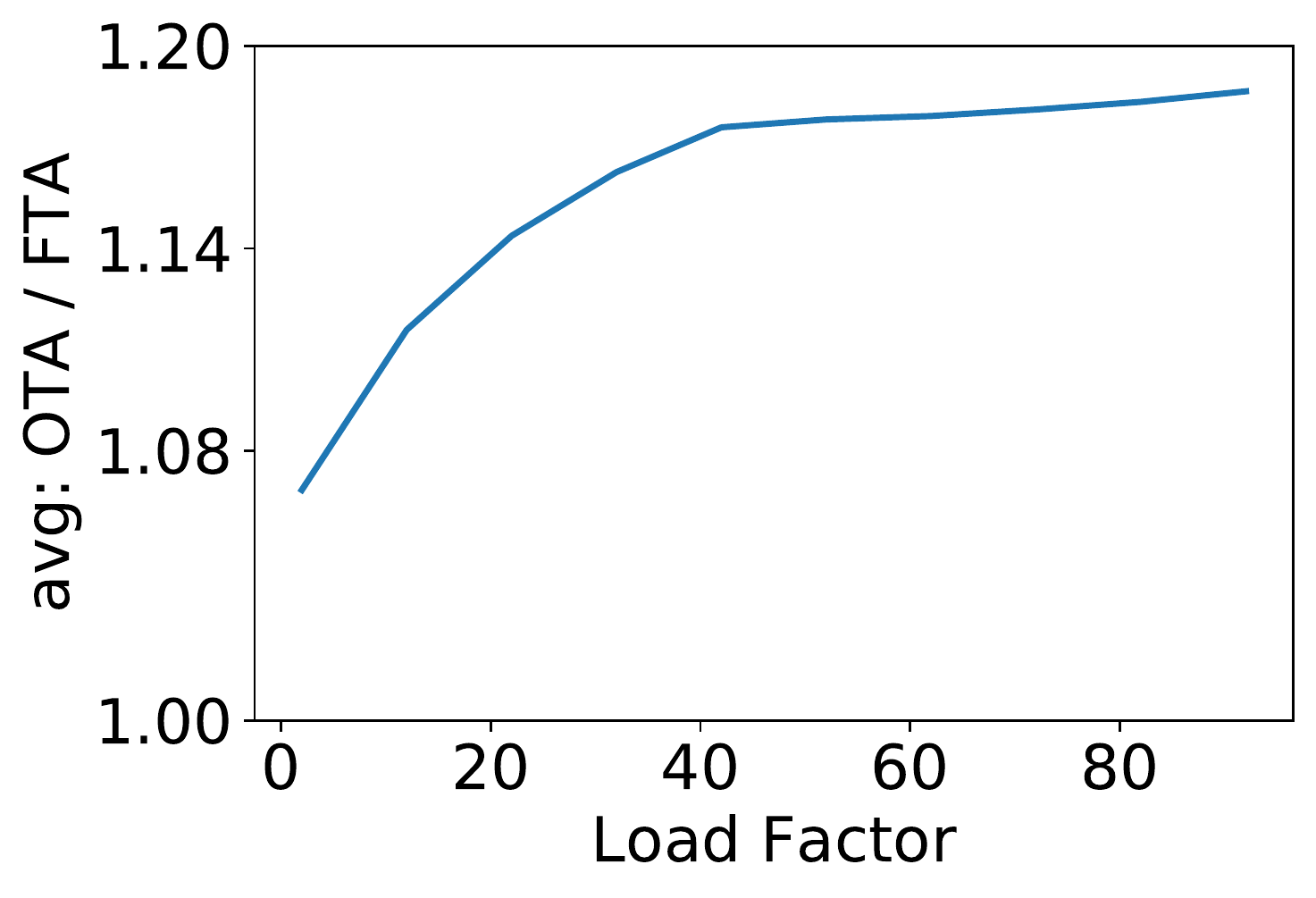}}
	
	\subfigure[Quadratic value function and low congestion]{\includegraphics[width=0.24\textwidth]{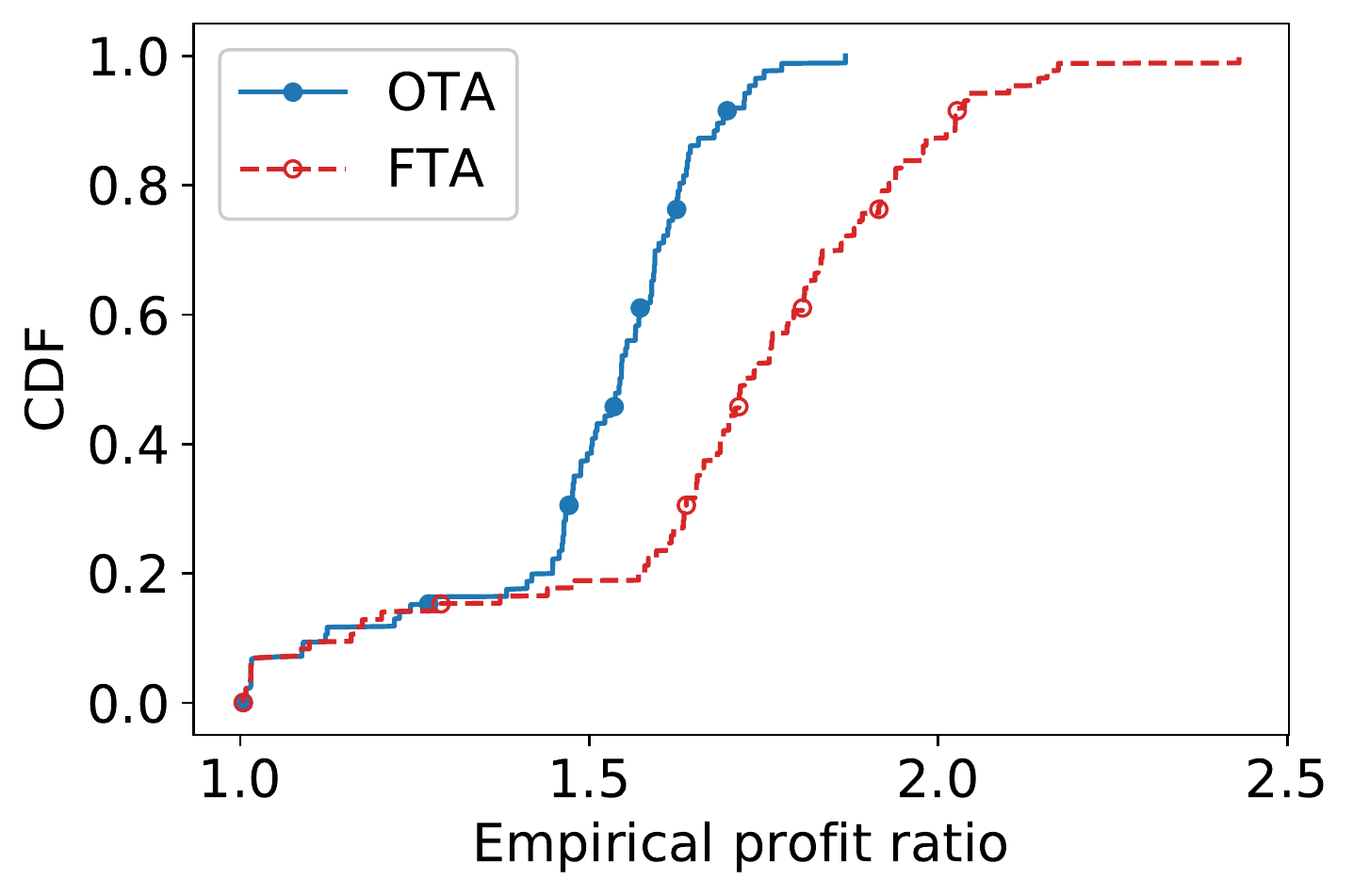}}
	\subfigure[Quadratic value function and medium congestion]{\includegraphics[width=0.24\textwidth]{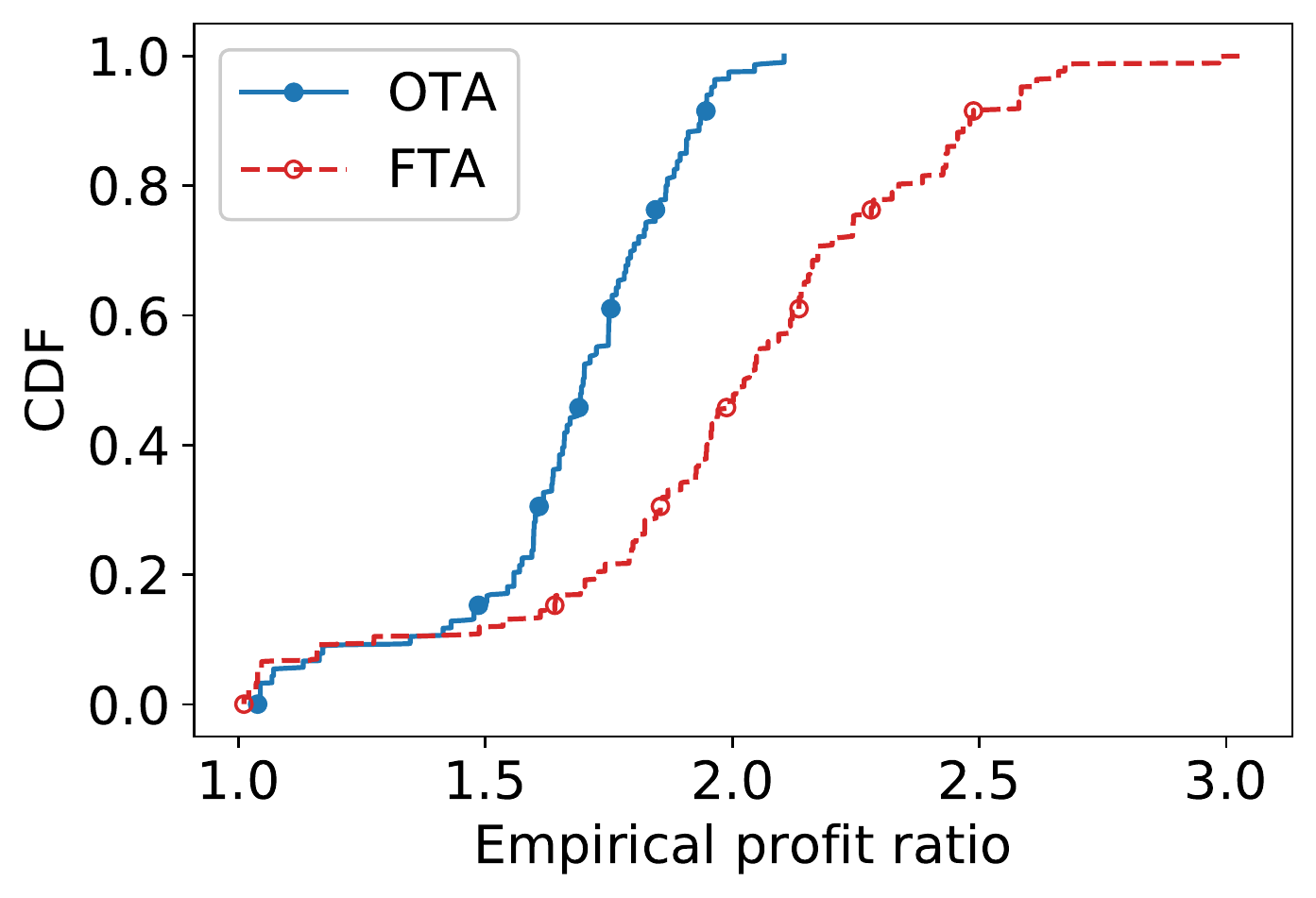}}
	\subfigure[Quadratic value function and high congestion]{\includegraphics[width=0.24\textwidth]{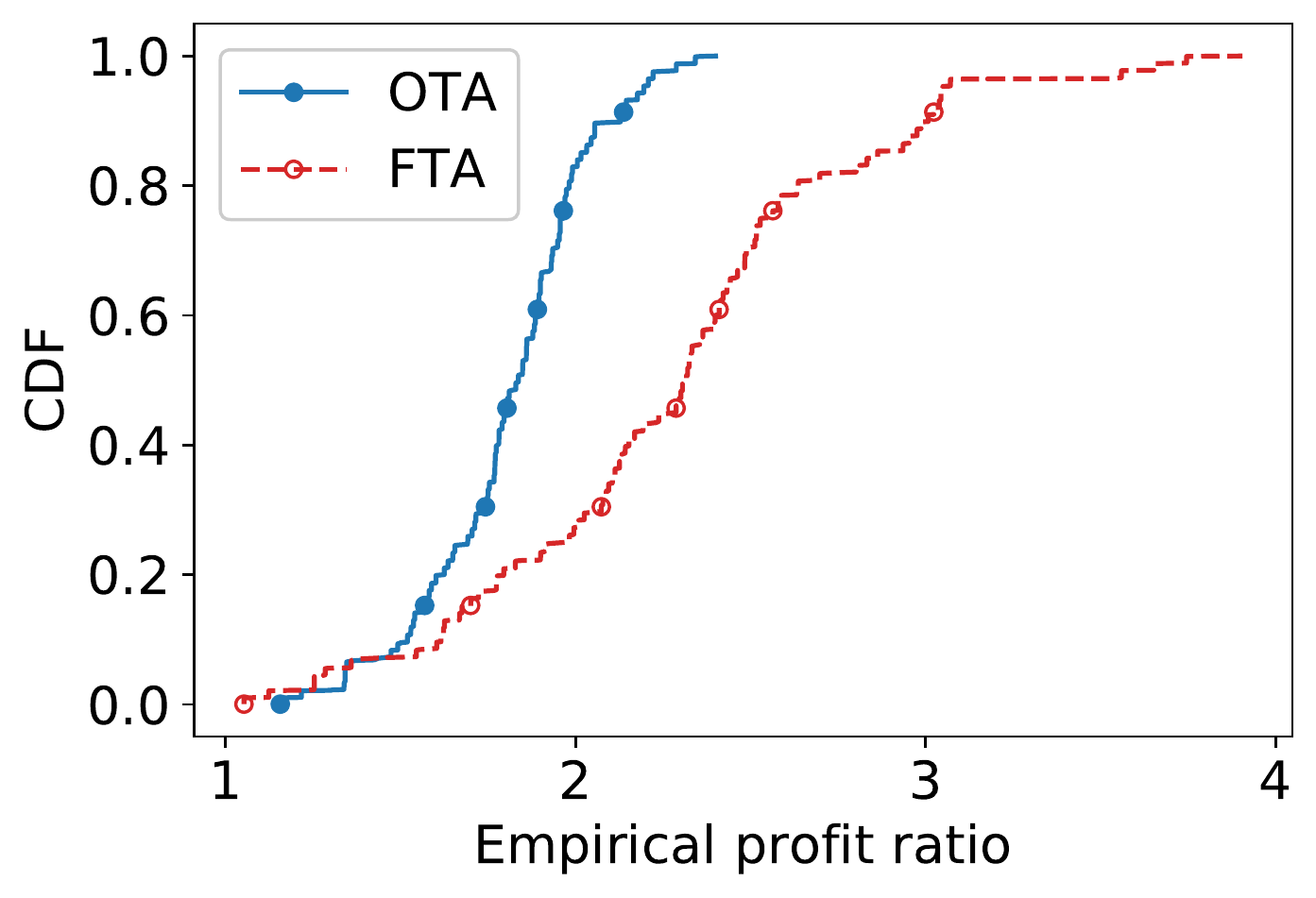}}
	\subfigure[Quadratic value function and varying congestion]{\includegraphics[width=0.24\textwidth]{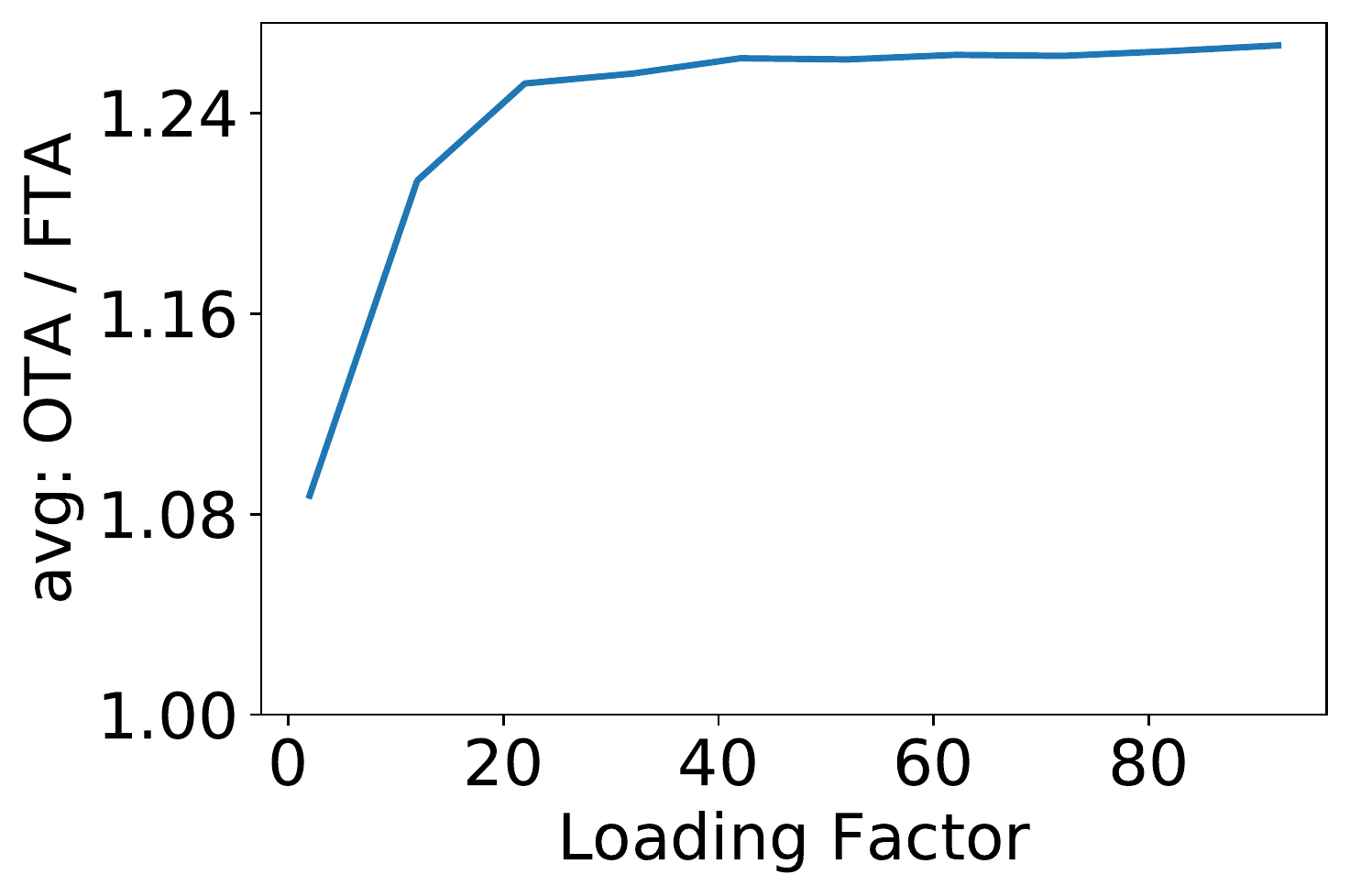}}
	\caption{The CDF of empirical profit ratios of \ota and \textsf{FTA} in low, medium, and high congestion settings with linear (a--d) or quadratic (e--h) value functions are shown. The results highlight that when the value function is linear, \ota improves the worst-case profit ratio over \textsf{FTA} by 10.0\%, 33.3\%, and 44.4\% for low, medium, and high congestion settings, respectively. When value functions are quadratic, these improvements are 23.2\%, 30.6\%, and 38.5\%. Figures~(d) and~(h) show the average improvement of \ota as compared to \textsf{FTA} as the load factor varies, highlighting a nearly 20\% improvement in high load settings when value functions are linear and 25\% when value functions are quadratic.}
	\label{fig:results_1}
\end{figure*}

\textbf{Experimental Results.} Our focus is on the competitive ratio. To illustrate the improvement in the worst-case performance, Figure~\ref{fig:results_1} demonstrates the cumulative distribution function (CDF) of the empirical profit ratios of \ota, our proposed algorithm, and \textsf{FTA} in low, medium, and high congestion levels. 
When value functions are linear, the results show that the profit ratio of our algorithm in the low, medium, and high congestion levels is bounded by 1.7, 2.0, and 2.5 while the maximum profit ratio of \textsf{FTA} is 1.9, 3.0, and 4.5. This represents a decrease of 10.0\%, 33.3\%, and 44.4\%, respectively, thus we see nearly a factor of 2 improvement in the worst case when congestion is high. While the adaptive threshold was designed with the worst-case in mind, we also see an improvement in the average profit of 6.7\%, 11.8\%, and 15.8\% in low, medium, and high congestion levels, respectively. In Figure~2(d), we report the percentage of improvement in the average profit ratio as the congestion level increases. The result shows that, as the system becomes more congested, the improvement of \ota grows since the value of scheduling increases with congestion.
In practice, EVs may have diminishing returns, which can be characterized by quadratic value functions.
In Figures~2(e)--2(h), we show the empirical profit ratios when value functions are quadratic.
Comparing the linear and quadratic settings, we observe the average improvement of \ota compared to \textsf{FTA} under high congestion is improved from nearly 20\% to 25\%. This highlights the importance of taking into account the non-linearity of value functions in \ota compared to \textsf{FTA}, which only uses the upper and lower bounds of the marginal value function.    

Finally, we performed experiments to understand the impact of heterogeneity in the values of drivers. We considered 9 classes of arrivals, with non-i.i.d. value distributions.  The mean of the classes differ by 20.  Heterogeneity leads to an increase in the improvement of our algorithm over \textsf{FTA}.  The resulting improvements are 27.4\%, 68.9\%, and 56.8\% for the worst case and 10.7\%, 17.9\%, and 20.7\% for the average case in low, medium, and high congestion, respectively.

\section{Concluding remarks}

Motivated by EV charging, we have introduced a general form of a fractional online multiple knapsack problem that includes heterogeneous rate constraints and provides a unification of a range of variants of both online knapsack and one-way trading.  Our main result provides a near-optimal algorithm for the general problem, as well as results for special cases corresponding to a number of variants that have received attention in recent years. In all cases, we either match or improve on state-of-the-art results while also including features, such as rate-constraints, that were not included in prior work.  

The key to our results is a new analytic technique called instance-dependent online primal-dual analysis, which provides a systematic way to design threshold functions for \ota, something that was previously more art than science.  The approach exposes a novel and powerful connection between the design of algorithms and the identification of worst-case instances.  We expect this technique to be applicable beyond the online knapsack and one-way trading problems, and an important line of future work is to understand the breadth of online algorithms problems where the identification of worst-case instances can systematically guide the design of algorithms.

Beyond exploring the impact of the analytic approach, another important line of future work is to explore the application of the new algorithm proposed here.  We briefly highlight an application in EV charging, but more work is needed before real-world deployment. Moreover, applications to cloud scheduling and geographical load balancing will be exciting to pursue.

\section{Acknowledgment}
Bo Sun and Danny H.K. Tsang acknowledge the support received from the Hong Kong Research Grant Council (RGC) General Research Fund (Project 16202619 and Project 16211220). 
Ali Zeynali and Mohammad Hajiesmaili's research is supported by NSF CNS-1908298. 
Tongxin Li's research is supported by NSF grants (CPS ECCS 1932611 and CPS ECCS 1739355).
Adam Wierman acknowledges the support received from NSF grants (AitF-1637598 and CNS-1518941).

\bibliographystyle{plain}  
\bibliography{references}

\newpage
\appendix

\section{Proofs}
\subsection{Proof of Proposition~\ref{pro:conjugate-function}}\label{app:proof-conjugate-function}

We begin by proving property (i) in Proposition~\ref{pro:conjugate-function}.
Let $\la \ge \la' \ge 0$ and denote by $\tilde{y}_n$ the optimal solution that maximizes the conjugate problem~\eqref{eq:conjugate-function} given $\la$. We then have
\begin{align*}
	h_n(\la) = g_n(\tilde{y}_n) - \la \tilde{y}_n  \le g_n(\tilde{y}_n) - \la' \tilde{y}_n \le \max_{0\le y_n\le D_n} g_n(y_n) - \la' y_n = h_n(\la').
\end{align*}
Thus, $h_n(\la)$ is a non-increasing function.

Note that the threshold function $\phi$ is discontinuous at $C$ since $\phi(w) = +\infty, w\in(C,+\infty)$ by definition. When $\phi(C) \ge U$, there must exist a utilization level $\bar{w} \le C$ and $\phi(\bar{w}) = U$. Since $g_n' \le U$, we can have $w^{(n)} + y_n \le \bar{w}$. Consequently, the derivative of the integral function
\begin{align*}
	\Phi(y_n) :=  \smallint_{w^{(n)}}^{w^{(n)}+y_n}\phi(u)du
\end{align*}
is continuous and non-decreasing, and hence this integral function is convex when $y_n\leq\bar{w}-w^{(n)}$. 
Thus, when $\phi(C) \ge U$, the pseudo-utility maximization problem~\eqref{p:utility-maximization-got} is a convex optimization problem, and its optimal solution can be determined by the KKT conditions.
Part of the KKT conditions is given by
\begin{align*}
	g_n'(y_n^*) - \phi(w^{(n+1)}) - \kappa_n^* + \nu_n^* = 0, \quad \kappa_n^*(D_n - y_n^*) = 0, \quad \nu_n^* y_n^* = 0,
\end{align*}
where $y_n^*$ and $\{\kappa_n^*,\nu_n^*\}$ are the optimal primal and dual variables, and $\kappa_n^*$ and $\nu_n^*$  correspond to the constraints $y_n \le D_n$ and $y_n \ge 0$, respectively. 
We next can verify the property (ii) in Proposition~\ref{pro:conjugate-function} by showing that $y_n^*$ maximizes the conjugate problem~\eqref{eq:conjugate-function} given $\la = \phi(w^{(n+1)})$. 

(i) When $g_n'(y_n^*) - \phi(w^{(n+1)}) >0$, we have $y_n^* = D_n$, i.e., $g_n'(D_n) > \phi(w^{(n+1)})$. In this case, we have $y_n^* = D_n = \argmax_{0\le y_n\le D_n} g_n(y_n) - \phi(w^{(n+1)}) y_n$. 

(ii) When $g_n'(y_n^*) - \phi(w^{(n+1)}) <0$, we have $y_n^* = 0$, i.e., $g_n'(0) < \phi(w^{(n+1)}) = \phi(w^{(n)})$. In this case, we have $y_n^* = 0 = \argmax_{0\le y_n\le D_n} g_n(y_n) - \phi(w^{(n+1)}) y_n$. 

(iii) When $g_n'(y_n^*) = \phi(w^{(n+1)})$, $y_n^*$ satisfies the KKT conditions of the conjugate optimization problem~\eqref{eq:conjugate-function}, given $\la = \phi(w^{(n+1)})$, i.e.,
\begin{align*}
	g_n'(y_n^*) - \phi(w^{(n+1)}) - \tilde{\kappa}_n + \tilde{\nu}_n = 0, \quad \tilde{\kappa}_n(D_n - y_n^*) = 0, \quad \tilde{\nu}_ny_n^* = 0,
\end{align*}
in which $\tilde{\nu}_n = \tilde{\mu}_n = 0$ since $g_n'({y}_n^*) - \phi(w^{(n+1)}) = 0$. Therefore, $y_n^* \in \argmax_{0\le y_n\le D_n} g_n(y_n) - \phi(w^{(n+1)}) y_n$.

Thus, we have $h_n(\phi(w^{(n+1)})) = g_n(y_n^*) - \phi(w^{(n+1)}) y_n^*, \forall n\in\caln$.

\subsection{Proof of Corollary~\ref{lem:threshold-varaint1-got}}\label{app:proof-variant1-got-threshold}

The dual of the offline problem~\eqref{p:varaint1-got} is
$$\min_{\la \ge 0} \sum\nolimits_{n\in\caln} h_n(\la) + (\la + L) C,$$
where $h_n(\la) = \max_{0\le y_n\le D_n} g_n(y_n) - (\la+L)y_n$. 
The pseudo-utility maximization problem of this variant is the same as \got, which is given by \eqref{p:utility-maximization-got}. 
We set the feasible dual solution as $\hat{\la} = \phi(w^{(N+1)}) - L$. 
Based on weak duality, we have
\begin{subequations}
	\begin{align*}
		\opt(\cali) &\le \sum\nolimits_{n\in\caln} h_n(\phi(w^{(N+1)}) - L) +  \phi(w^{(N+1)})C\\
		&\le\sum\nolimits_{n\in\caln} \left[g_n(y_n^*) - \phi(w^{(n+1)}) y_n^*\right] +  \phi(w^{(N+1)})C\\
		&\le \sum\nolimits_{n\in\caln}g_n(y_n^*) + (C - w^{(N+1)})L  +  \phi(w^{(N+1)})C - \smallint_{0}^{w^{(N+1)}}\phi(u)du - (C- w^{(N+1)})L\\
		&\le \alpha \left[\sum\nolimits_{n\in\caln}g_n(y_n^*) + (C - w^{(N+1)})L \right] = \alpha \alg(\cali).
	\end{align*}
\end{subequations}
When $\phi(C) \ge U$, applying Proposition~\ref{pro:conjugate-function} gives the second inequality. 
The last inequality holds if the threshold function further satisfies the following differential equation
\begin{align*}
	\phi(w) C \le \alpha {\smallint_{0}^w} \phi(u) du + \alpha (C - w)L, w\in[0,C].
\end{align*}

Applying Gronwall's Inequality to above equation, we obtain 
\begin{align*}
	\phi(w)\le \alpha L - \frac{\alpha L w}{C} + \frac{\alpha}{C} {\smallint_{0}^w} \phi(u) du \le L + (\alpha L - L) e^{\alpha w/C}.
\end{align*}
Note that $\phi(C) \ge U$ gives $ U\le \phi(C) \le L + (\alpha L - L) e^{\alpha}$. When all inequalities hold with equality, the competitive ratio $\alpha$ is minimized and is given by the solution of the equation $\alpha = \ln\frac{U-L}{\alpha L - L}$, and the threshold function $\phi$ is given by~\eqref{eq:threshold-got-variant1}.

\subsection{Proof of Corollary~\ref{lem:varaint1-got}}\label{app:proof-varaint1-got-opt}

In this variant of \got, we can have the following necessary condition for the existence of an $\alpha$-competitive online algorithm based on the continuously non-decreasing instance (see Definition~\ref{dfn:instance}).

\begin{claim}
	\label{pro:necessary-variant1-got}
	If there exists an $\alpha$-competitive online algorithm for Variant 1 of \got, there must exist a utilization function $\psi(p): [L,U]\to[0,C]$ that is non-decreasing and satisfies 
	\begin{align}\label{eq:necessary-variant1-got}
		\begin{cases}
			\psi(p)p - \smallint_{L}^p\psi(u)du + [C - \psi(p)]L \ge pC/\alpha, p\in[L,U]\\
			\psi(U) \le C
		\end{cases}.
	\end{align} 
\end{claim}

To prove this, note that under instance $\cali_p$, we have
\begin{align*}
	&\opt(\cali_p) = pC, \\
	&\alg(\cali_p) = \psi(L)L + \smallint_{L}^p u d\psi(u) + (C - \psi(p))L = \psi(p)p - \smallint_{L}^p\psi(u)du + [C - \psi(p)]L.
\end{align*}
Since any $\alpha$-competitive online algorithm must satisfy $\alg(\cali_p) \ge \opt(\cali_p)/\alpha$, this gives the differential equation in \eqref{eq:necessary-variant1-got}. The utilization function cannot exceed the capacity so we have the boundary condition $\psi(U)\le C$.

The differential equation in~\eqref{eq:necessary-variant1-got} holds when $\psi(p) = 0,  p\in[L,\alpha L]$. We can then apply Gronwall's Inequality to \eqref{eq:necessary-variant1-got} and obtain
\begin{align*}
	\psi(p)\ge \frac{C}{\alpha} \frac{p}{p-L} - \frac{C L}{p- L} + \frac{1}{p-L}{\int_{\alpha L}^{p}} \left[\frac{C}{\alpha} \frac{u}{u-L} - \frac{C L}{u- L}\right] \frac{p-L}{u-L}du = \frac{C}{\alpha} \ln\frac{p-L}{\alpha L - L}.
\end{align*}
Since $\psi(U) \le C$, we have $C \ge \psi(U) \ge \frac{C}{\alpha} \ln\frac{U-L}{\alpha L - L}$, and those inequalities hold when all inequalities in ~\eqref{eq:necessary-variant1-got} are binding. Thus, a lower bound of the optimal competitive ratio is the solution of the equation $\alpha = \ln\frac{U-L}{\alpha L - L}$ and the corresponding utilization function is given by
\begin{align*}
	\psi^*(p) = 
	\begin{cases}
		0 & p \in[L,\alpha L]\\
		\frac{C}{\alpha}\ln\frac{p-L}{\alpha L - L} & p\in [\alpha L,U]
	\end{cases}.
\end{align*}
Since this lower bound can be achieved by $\ota_\phi$ with the threshold function given in Corollary~\ref{lem:threshold-varaint1-got}, the optimal competitive ratio $\alpha^*$ for Variant 1 of \got is the solution of the equation $\alpha^* =  \ln\frac{U-L}{\alpha^* L - L}$.

\subsection{Proof of Corollary~\ref{lem:varaint2-got}}

We begin by proving a sufficient condition for ensuring an $\alpha$-competitive \ota for Variant 2 of \got.

\begin{claim}\label{lem:sufficient-variant2-got}
	Under the conditions (i) and (ii) in Assumption~\ref{ass:value-function} and Assumption~\ref{ass:relaxed}, the $\ota_\phi$ for Variant 2 of \got is $\alpha$-competitive if the threshold function is in the form of 
	\begin{align*}
		\phi(w)=
		\begin{cases}
			\varphi_1(w)&  w\in[0, \beta)\\
			\varphi_2(w) & w\in[\beta, C]
		\end{cases},
	\end{align*}
	where $\beta \in[0, C]$ is a utilization threshold and $\phi(w)$ satisfies the following conditions: 
	
	(i) $\varphi_1(w)$ is a non-decreasing differentiable function that satisfies
	\begin{align}\label{eq:variant2-got-sufficient-cond-segment1}
		\begin{cases}
			\varphi_1(w) C \le \smallint_{0}^{w}\varphi_1(u)du + (\alpha - 1)\frac{L}{c}w, w\in[0, \beta),\\
			\varphi_1(0) = 0, \varphi_1(\beta) = L/c .
		\end{cases}
	\end{align}
	
	(ii) $\varphi_2(w)$ is a non-decreasing differentiable function that satisfies
	\begin{align}\label{eq:sufficient-condition-relaxed-segment2}
		\begin{cases}
			\varphi_2(w) C \le \alpha \smallint_{0}^{w}\phi(u)du, w\in[\beta,C],\\
			\varphi_2(\beta) = L/c, \varphi_2(C) \ge U.
		\end{cases}
	\end{align} 
\end{claim}

The derivation of the above sufficient condition for $w\in[\beta, C]$ is the same as that of Case II in \got by changing $L$ to $L/c$. 
Different from Case I of \got, the marginal value function of this case is not strictly lower bounded. Thus, setting the threshold function in $[0,\beta)$ to a flat segment cannot ensure all items are accepted to their sizes and the argument in Case I fails. To handle this relaxed assumption, we can design the threshold function as a non-decreasing function $\varphi_1(w)$. Following the OPD approach in Case II of \got, we can build the upper bound of $\opt(\cali)$ until the equation~\eqref{eq:proof-got-ineq1}. Instead of~\eqref{eq:proof-got-sufficient-conds}, a smaller competitive ratio can be achieved by enforcing a stringent sufficient condition~\eqref{eq:variant2-got-sufficient-cond-segment1} since $\varphi_1(w) \le L/c$. Then we have
\begin{align*}
	\opt(\cali) \le & \sum\nolimits_{n\in\caln}g_n(y_n^*) + (\alpha -1) \frac{L}{c}w^{(N+1)} \le \sum\nolimits_{n\in\caln}g_n(y_n^*) + (\alpha -1) \sum\nolimits_{n\in\caln}g_n(y_n^*) = \alpha \alg(\cali).
\end{align*}
Based on Assumption~\ref{ass:relaxed} and the concavity of $g_n(\cdot)$, we have $L/c \le g_n(D_n)/D_n \le g_n(y_n^*)/y_n^*$. 
Then the second inequality is given by $w^{(N+1)}L/c = \sum\nolimits_{n\in\caln}y_n^ *L/c \le \sum\nolimits_{n\in\caln} g_n(y_n^*)$. 

Solving $\varphi_1$ and $\varphi_2$ by binding all inequalities in the sufficient conditions (i) and (ii), we can obtain the threshold function~\eqref{eq:variant2-got-threshold-fun} and the corresponding competitive ratio in Corollary~\ref{lem:varaint2-got}.

\subsection{Proof of Lemma~\ref{lem:fomkp-sufficient-aggregate}}\label{app:proof-fomkp-aggregate}

To begin, we rewrite the offline formulation of the \fomkp with aggregate value functions as follows:
\begin{subequations}\label{p:FOMKP-aggregate}
	\begin{align}
		\max_{0\le x_n \le D_n, y_{nm} \ge 0} &\quad \sum\nolimits_{n\in\caln} g_n(x_n),\\
		\label{eq:FOMKP-primal-energy}
		{\rm s.t.} &\quad \sum\nolimits_{m\in\calm} y_{nm} \ge x_n, \forall n\in\caln, \quad (\mu_n)\\
		\label{eq:FOMKP-primal-capacity}
		&\quad \sum\nolimits_{n\in\caln}y_{nm} \le C_m, \forall
		m\in\calm, \quad (\la_m)\\
		\label{eq:FOMKP-primal-rate}
		&\quad y_{nm} \le Y_{nm}, \forall n\in\caln, m\in\calm, \quad (\gamma_{nm})
	\end{align}
\end{subequations}
where $x_n$ is the assigned aggregate fraction of item $n$. 
Similarly, the pseudo-utility maximization problem~\eqref{p:utility-maximization} in $\ota_\phi$ can be rewritten as
\begin{subequations}
	\label{p:FOMKP-utility-max}
	\begin{align}
		\max_{0\le x_n\le D_n, y_{nm}} \quad & g_n(x_n) - \sum\nolimits_{m\in \calm}  \smallint\nolimits_{w_m^{(n)}}^{w_m^{(n)} + y_{nm}}\phi_{m}(u)du\\
		\label{eq:FOMKP-utility-max-energy}
		{\rm s.t.}\quad& \sum\nolimits_{m\in\calm}y_{nm} \ge x_n, \quad (\mu_n) \\
		\label{eq:FOMKP-utility-max-rate}
		& 0\le y_{nm}\le Y_{nm},\forall n\in\caln, m\in\calm. \quad (\xi_{nm}, {\gamma}_{nm})
	\end{align}
\end{subequations}

The dual of the offline problem~\eqref{p:FOMKP-aggregate} can be derived as
\begin{subequations}\label{p:dual-FOMKP-aggregate}
	\begin{align}\label{eq:dual-FOMKP-aggregate-obj}
		\min_{\la_m\ge0,\mu_n\ge0,\gamma_{nm}\ge0}\quad& \sum\nolimits_{n\in\caln} h_n(\mu_n) + \sum\nolimits_{m\in\calm} \la_m C_m + \sum\nolimits_{n\in\caln}\sum\nolimits_{m\in\calm} \gamma_{nm} Y_{nm} \\
		{\rm s.t.}\quad& \la_m \ge \mu_n - \gamma_{nm}, \quad \forall n\in\caln,m\in\calm,
	\end{align}
\end{subequations}
where $h_n(\mu_n) = \max_{0\le x_n\le D_n} g_n(x_n) - \mu_n x_n$ is the conjugate function of $g_n(\cdot)$, and $\bmu:=\{\mu_n\}_{n\in\caln}$, $\bla:=\{\la_m\}_{m\in\calm}$, and $\bga:=\{\gamma_{nm}\}_{n\in\caln,m\in\calm}$ are the dual variables associated with constraints~\eqref{eq:FOMKP-primal-energy},~\eqref{eq:FOMKP-primal-capacity}, and~\eqref{eq:FOMKP-primal-rate}, respectively.
Let $\dual(\bmu, \bla, \bga)$ denote the dual objective~\eqref{eq:dual-FOMKP-aggregate-obj}.

To build the connection between the online solution and the dual objective in the \opd analysis, we need the following proposition, which is a general version of Proposition~\ref{pro:conjugate-function}. 
\begin{pro}\label{pro:conjugate-function-fomkp}
	When $\phi_m(C_m)\ge U, \forall m\in\calm$, the conjugate function $h_n(\mu_n)$ satisfies $h_n(\mu^*_n) = g_n(x_n^*) - \mu^*_n x_n^*$, where $x_n^*$ and $\mu_n^*$ are the optimal primal and dual solutions of the problem~\eqref{p:FOMKP-utility-max}.
\end{pro}
\begin{proof}
	When $\phi_m(C_m) \ge U, \forall m\in\calm$, the pseudo-utility maximization problem~\eqref{p:FOMKP-utility-max} is a convex optimization problem. 
	Part of its KKT conditions is given by
	\begin{align*}
		&g_n'(x_n^*) - \mu_n^* - \kappa_n^* + \nu_n^* = 0, \quad \kappa_n^*(D_n - x_n^*) = 0, \quad \nu_n^* x_n^* = 0,
	\end{align*}
	where $\mu_n^*$, $\kappa_n^*$, and $\nu_n^*$ are the optimal dual variables associated with the constraint~\eqref{eq:FOMKP-utility-max-energy}, $x_n \le D_n$, and $x_n \ge 0$.
	We can then follow the same arguments as those (i)-(iii) in the proof of Proposition~\ref{pro:conjugate-function} in Appendix~\ref{app:proof-conjugate-function} by just replacing $\phi(w^{(n+1)})$ with $\mu_n^*$.
\end{proof}

Let $w_m^{(N+1)}:= w_m^{(N+1)}(\cali)$ denote the final utilization of the knapsack $m$ after executing instance $\cali$ by $\ota_\phi$. The set of all instances $\Omega$ can be divided into three families $\Omega^1$, $\Omega^2$, and $\Omega^3$.
In particular, $\Omega^1 := \{\cali: 0 \le w_m^{(N+1)} < \beta_m, \forall m\in\calm \}$ and $\Omega^2 := \{\cali: \beta_m\le w_m^{(N+1)} \le C_m, \forall m\in\calm \}$ contain the instances whose final utilizations of all knapsacks are below and above their utilization thresholds, respectively. 
Excluding these two families, the remaining instances form $\Omega^3 := \Omega\setminus(\Omega^1\cup \Omega^2)$, in which some knapsacks have final utilizations below the utilization thresholds and the others' final utilizations are above the thresholds. We now treat these three cases separately.

{\bf Case I: $\cali\in\Omega^1$.} The threshold functions of all knapsacks are on the flat segment, which implies that the marginal cost of packing items into all knapsacks are at lowest price $L$. Thus, in this case all items are packed up to their sizes by maximizing the pseudo-utility in $\ota_\phi$. The offline optimal solution also accepts all items and hence is the same as the online solution. So, we have $\opt(\cali)/\alg(\cali) = 1 \le \alpha, \forall \cali \in\Omega^1$.

{\bf Case II: $\cali\in\Omega^2$.} 
The adversary can add one more item for each knapsack. The new item for knapsack $m$ is with size $C$ and marginal value $\phi_m(w_m^{(N+1)})$.  
Under this created worst-case instance, all knapsack capacities are occupied in the offline solution while the online solution keeps the same.
In this case, we can reply on the dual objective $\dual(\bmu, \bla, \bga)$ in~\eqref{eq:dual-FOMKP-aggregate-obj}.
A feasible dual solution can be constructed as
\begin{align*}
	\hat{\la}_m = \phi_m(w_m^{(N+1)}), \quad \hat{\mu}_n = \mu_n^*, \quad \hat{\gamma}_{nm} = {\gamma}_{nm}^* = 
	\begin{cases}
		\mu_n^* - \phi_m(w_m^{(n+1)}) & y_{nm}^* > 0\\
		0 & y_{nm}^* = 0
	\end{cases},
\end{align*}
where $\mu_n^*$ and $\gamma^*_{nm}$ are the optimal dual variables of the pseudo-utility maximization problem~\eqref{p:FOMKP-utility-max} associated with constraints~\eqref{eq:FOMKP-utility-max-energy} and~\eqref{eq:FOMKP-utility-max-rate}.
We first show that the dual variables are feasible. $\hat{\la}_m, \hat{\mu}_n, \hat{\gamma}_{nm} \ge0$ can be immediately observed. The dual constraint can be checked by
\begin{align}\label{eq:opd-feasibility}
	\hat{\mu}_n - \hat{\gamma}_{nm}
	= \phi_m(w_m^{(n+1)}) - {\xi}_{nm}^* 
	\le \phi_m(w_m^{(n+1)}) \le \phi_m(w_m^{(N+1)}) = \hat{\la}_m, \forall n\in\caln, m\in\calm,
\end{align}
where the first equality is the KKT condition of the problem~\eqref{p:FOMKP-utility-max} and $\xi_{nm}^* \ge 0$ is the optimal dual variable associated with the constraint $y_{nm} \ge 0$.
The third inequality holds since $w_m^{(n+1)} \le w_m^{(N+1)}$ and $\phi_m$ is a non-decreasing function. 
Then we can have
\begin{subequations}
	\begin{align}\label{eq:opd-fomkp-weak-duality}
		\opt(\cali) &\le 
		\sum\nolimits_{n\in\caln} \left[g_n(x_n^*) - \mu_n^* x_n^* \right] + \sum\nolimits_{m\in\calm} \phi_m(w_m^{(N+1)})C_m + \sum\nolimits_{n\in\caln}\sum\nolimits_{m\in\calm} \gamma_{nm}^*Y_{nm}\\
		\label{eq:opd-fomkp-kkt}
		&\le \sum\nolimits_{n\in\caln}g_n(x_n^*) + \sum\nolimits_{m\in\calm} \left[\phi_m(w_m^{(N+1)})C_m - \smallint_{0}^{w_m^{(N+1)}} \phi_m(u)du\right]\\
		\label{eq:opd-fomkp-sufficient}
		&\le \sum\nolimits_{n\in\caln}g_n(x_n^*) + (\alpha - 1)\sum\nolimits_{m\in\calm} \smallint_{0}^{w_m^{(N+1)}} \phi_m(u)du\\
		\label{eq:opd-fomkp-lower-bound-primal}
		&\le \sum\nolimits_{n\in\caln}g_n(x_n^*) + (\alpha - 1)\sum\nolimits_{n\in\caln}g_n(x_n^*) = \alpha \alg(\cali).
	\end{align}
\end{subequations}
Applying weak duality and Proposition~\ref{pro:conjugate-function-fomkp} gives inequality~\eqref{eq:opd-fomkp-weak-duality}.
Based on the KKT conditions of the problem~\eqref{p:FOMKP-utility-max}, $\mu_n^*(\sum_{m\in\calm}y_{nm}^* - x_n^*) = 0$ and $\gamma_{nm}^*(Y_{nm} - y_{nm}^*) = 0$, we have
\begin{align*}
	-\sum\nolimits_{n\in\caln} \mu_n^* x_n^* + \sum\nolimits_{n\in\caln}\sum\nolimits_{m\in\calm} \gamma_{nm}^*Y_{nm} 
	& = -\sum\nolimits_{n\in\caln}\sum\nolimits_{m\in\calm}[\mu_n^*- \gamma_{nm}^*]y_{nm}^*\\
	&=-\sum\nolimits_{n\in\caln}\sum\nolimits_{m\in\calm}\phi_m(w_m^{(n+1)})y_{nm}^*.
\end{align*}
Combining with $\phi_m(w_m^{(n+1)})y_{nm}^* \ge \smallint_{w_m^{(n)}}^{w_m^{(n+1)}}\phi_m(u)du$, the inequality~\eqref{eq:opd-fomkp-kkt} holds.
If the differential equation~\eqref{eq:fomkp-aggregate-sufficient} in Lemma~\ref{lem:fomkp-sufficient-aggregate} holds, the inequality~\eqref{eq:opd-fomkp-sufficient} holds.
Finally, we can have the inequality~\eqref{eq:opd-fomkp-lower-bound-primal} by observing that $g_n(x_n^*) \ge \sum\nolimits_{m \in \calm}  \smallint\nolimits_{w_m^{(n)}}^{w_m^{(n+1)}}\phi_{m}(u)du$ for $n\in\caln$ based on the problem~\eqref{p:FOMKP-utility-max}. Thus, in this case, we have $\opt(\cali)/\alg(\cali) \le \alpha, \forall \cali\in\Omega^2$ if the sufficient conditions in Lemma~\ref{lem:fomkp-sufficient-aggregate} are satisfied.

{\bf Case III: $\cali\in\Omega^3$.}
Let $\calm^1:=\{m\in\calm: 0 \le w_m^{(N+1)} < \beta_m \}$ and $\calm^2:=\{m\in\calm: \beta_m \le w_m^{(N+1)} < C_m \}$ denote the subsets of knapsacks, whose final utilizations are below and above the utilization thresholds, respectively.
The key difference between Case II and Case III is that the knapsacks in $\calm^1$ may not be fully occupied in the offline solution under the worst-case instance in $\Omega^3$.
This is because the total amount of items, which can be packed into $\calm^1$, is limited by $\sum_{m\in\calm^1}w_{m}^{(N+1)} + \sum_{m\in\calm^2} \beta_m$.

Based on this understanding of the worst-case instance, we can add the following constraint to the offline formulation~\eqref{p:FOMKP-aggregate}
\begin{align}\label{eq:opd-new-constraint}
	\sum\nolimits_{n\in\caln}\sum\nolimits_{m\in\calm^1}y_{nm} \le \sum\nolimits_{m\in\calm^1} w_m^{(N+1)} + \sum\nolimits_{m\in\calm^2} \beta_m.
\end{align}

The dual problem of the new offline problem can be stated as
\begin{subequations}\label{eq:opd-dual-caseIII}
	\begin{align}\label{eq:opd-dual-caseIII-obj}
		\min_{\la_m\ge0,\mu_n\ge0,\gamma_{nm}\ge0, \eta \ge 0}\quad& \dual(\bmu, \bla, \bga) + \eta(\sum\nolimits_{m\in\calm^1} w_m^{(N+1)} + \sum\nolimits_{m\in\calm^2} \beta_m ) \\
		\label{eq:opd-dual-caseIII-cont1}
		{\rm s.t.}\quad& \la_m \ge \mu_n - \gamma_{nm}, \quad \forall n\in\caln,m\in\calm^2,\\
		\label{eq:opd-dual-caseIII-cont2}
		&\la_m \ge \mu_n - \gamma_{nm} - \eta, \quad \forall n\in\caln, m\in\calm^1,
	\end{align}
\end{subequations}
where $\eta$ is the dual variable of the new constraint~\eqref{eq:opd-new-constraint}. 
We can construct a feasible dual solution as
\begin{align*}
	\hat{\la}_m =
	\begin{cases}
		\phi_m(w_m^{(N+1)}) & m\in\calm^2 \\
		0 & m\in\calm^1
	\end{cases},\quad
	\hat{\mu}_n = \mu_n^*,\quad
	\hat{\gamma}_{nm} &= {\gamma}_{nm}^*, \quad
	\hat\eta = L.
\end{align*}

Based on equation~\eqref{eq:opd-feasibility} in Case II, the constructed dual solution satisfies the constraint~\eqref{eq:opd-dual-caseIII-cont1}. Furthermore, the constraint~\eqref{eq:opd-dual-caseIII-cont2} can be checked by observing
\begin{align}
	\hat{\mu}_n - \hat{\gamma}_{nm}&\le \phi_m(w_m^{(N+1)})= L = \hat{\la}_m + \hat{\eta},n\in\caln, m\in\calm^1.
\end{align}

Applying the dual objective in~\eqref{eq:opd-dual-caseIII-obj} and KKT conditions of the problem~\eqref{p:FOMKP-utility-max}, we can have
\begin{align}\nonumber
	\opt(\cali) &\le 
	\sum\nolimits_{n\in\caln}g_n(x_n^*) + \sum\nolimits_{m\in\calm^2} [\phi_m(w_m^{(N+1)}) C_m -\smallint_{0}^{w_m^{(N+1)}} \phi_m(u)du + L \beta_m]\\
	\label{eq:opd-sufficient-caseIII}
	&{\le \sum\nolimits_{n\in\caln}g_n(x_n^*) + (\alpha - 1)\sum\nolimits_{m\in\calm^2}\smallint_{0}^{w_m^{(N+1)}} \phi_m(u)du} \le \alpha \alg(\cali),
\end{align}
where the inequality~\eqref{eq:opd-sufficient-caseIII} holds if $\phi$ satisfies the differential equation~\eqref{eq:fomkp-aggregate-sufficient} in Lemma~\ref{lem:fomkp-sufficient-aggregate}. 
Thus, we have $\opt(\cali)/\alg(\cali) \le \alpha, \forall \cali\in\Omega^3$.

In summary, the competitive ratio is $\alpha$ if $\phi$ satisfies the sufficient condition in  Lemma~\ref{lem:fomkp-sufficient-aggregate}.

\subsection{Proof of Theorem~\ref{thm:threshold-omkp-separable}}\label{app:proof-thm-fomkp-separable}

This proof proceeds much the same as that of Theorem~\ref{thm:threshold-omkp-aggregate}, but we now use the following counterpart to the sufficient conditions on the threshold functions of $\ota_\phi$.

\begin{lem}
	\label{lem:fomkp-sufficient-separable}
	Under Assumption~\ref{ass:value-function}, $\ota_\phi$ for the \fomkp with separable value functions is $\alpha$-competitive if the threshold function $\phi=\{\phi_m\}_{m\in\calm}$ is in the form of, $\forall m\in\calm$, 
	\begin{align*}
		\phi_m(w)=
		\begin{cases}
			L&  w\in[0,\beta_m)\\
			\varphi_m(w) & w\in[\beta_m, C_m]
		\end{cases},
	\end{align*}
	where $\beta_m \in [0,C_m]$ is a utilization threshold and $\varphi_m$ is a non-decreasing function, and $\phi_m$ satisfies
	\begin{align}\label{eq:fomkp-sufficient-separable}
		\begin{cases}
			\varphi_m(w)C_m \le \alpha \smallint_{0}^{w}\phi_m(u)du - Lw, \quad  w\in[\beta_m,C_m],\\
			\varphi_m(\beta_m) = L,\varphi_m(C_m) \ge U.
		\end{cases}
	\end{align} 
\end{lem}

\begin{proof}[Proof of Lemma~\ref{lem:fomkp-sufficient-separable}]
	
	The offline formulation of the \fomkp with separable value functions can be stated as 
	\begin{subequations}\label{p:fomkp-separable}
		\begin{align}\label{eq:fomkp-separable-obj}
			\max_{0\le y_{nm} \le Y_{nm}} &\quad \sum\nolimits_{n\in\caln}\sum\nolimits_{m\in\calm} g_{nm}(y_{nm}),\\
			\label{eq:fomkp-separable-demand}
			{\rm s.t.} &\quad \sum\nolimits_{m\in\calm} y_{nm} \le D_n, \forall n\in\caln, \\
			\label{eq:fomkp-separable-capacity}
			&\quad \sum\nolimits_{n\in\caln}y_{nm} \le C_m, \forall
			m\in\calm.
		\end{align}
	\end{subequations}
	
	The dual problem of this offline formulation can be derived as
	\begin{align}\label{p:dual-fomkp-separable}
		\min_{\la_m\ge0,\mu_n\ge0}\quad& \sum\nolimits_{n\in\caln}\sum\nolimits_{m\in\calm} h_{nm}(\mu_n + \la_m) + \sum\nolimits_{m\in\calm} \la_m C_m + \sum\nolimits_{n\in\caln} \mu_{n} D_n,
	\end{align}
	where $h_{nm}(\rho_{nm}) = \max_{0\le y_{nm}\le Y_{nm}} g_{nm}(y_{nm}) - \rho_{nm} y_{nm}$ is the conjugate function of $g_{nm}(\cdot)$, and $\mu_n$ and $\la_m$ are the dual variables that correspond to constraints~\eqref{eq:fomkp-separable-demand} and~\eqref{eq:fomkp-separable-capacity}, respectively.

	With separable value functions, the pseudo-utility maximization problem~\eqref{p:utility-maximization} can be rewritten as
	\begin{subequations}
		\label{p:FOMKP-utility-max-separable}
		\begin{align}
			\max_{0\le y_{nm}\le Y_{nm}} \quad & \sum\nolimits_{m\in\calm}g_{nm}(y_{nm}) - \sum\nolimits_{m\in \calm}  \smallint\nolimits_{w_m^{(n)}}^{w_m^{(n)} + y_{nm}}\phi_{m}(u)du\\
			\label{eq:FOMKP-utility-max-energy-separable}
			{\rm s.t.}\quad& \sum\nolimits_{m\in\calm}y_{nm} \le D_n. 
		\end{align}
	\end{subequations}
	
	We first connect the online solution and the dual objective through the following proposition.
	\begin{pro}\label{pro:conjugate-function-separable}
		The conjugate function $h_{nm}(\rho_{nm})$ satisfies
		
		(i) $h_{nm}(\rho_{nm})$ is a non-increasing function;
		
		(ii) when $\phi_m(C_m) \ge U, \forall m\in\calm$, $ h_{nm}(\phi_m(w_m^{(n+1)}) + \mu_n^*) = g_{nm}(y_{nm}^*) - (w_m^{(n+1)} + \mu_n^*) y_{nm}^*,$ where $y_{nm}^*$ and $\mu_n^*$ are the optimal primal and dual solutions of the problem~\eqref{p:FOMKP-utility-max-separable}, and $w_m^{(n+1)} = w_m^{(n)} + y_{nm}^*$.
	\end{pro}
	\begin{proof}[Proof of Proposition~\ref{pro:conjugate-function-separable}]
		The property (i) of the conjugate function can be shown in the same way as the proof of Proposition~\ref{pro:conjugate-function}.
		
		When $\phi_m(C_m) \ge U, \forall m\in\calm$, the pseudo-utility maximization problem~\eqref{p:FOMKP-utility-max-separable} is a convex optimization problem and part of its KKT conditions is given by
		\begin{align*}
			g_{nm}'(y_{nm}^*) - \phi_m(w_m^{(n+1)}) - \mu_n^* - \gamma_{nm}^* + \xi_{nm}^* = 0, \quad \gamma_{nm}^*(Y_{nm} - y_{nm}^*) = 0, \quad \xi_{nm}^* y_{nm}^* = 0,
		\end{align*}
		where $y_{nm}^*$, and $\{\mu_n^*,\gamma_{nm}^*, \xi_{nm}^*\}$ are the optimal primal and dual solutions. Based on this structure of KKT conditions, we can follow the same arguments in the proof of Proposition~\ref{pro:conjugate-function} in Appendix~\ref{app:proof-conjugate-function} to show that $y_{nm}^*$ maximizes the conjugate optimization problem  given $\rho_{nm} = \phi_m(w_m^{(n+1)}) + \mu_n^*$. 
	\end{proof}

	We proceed to derive the sufficient condition on the threshold function for the \fomkp with separable functions using the instance-dependent \opd analysis. 
	The set of instances $\Omega$ is divided into three families $\Omega^1$, $\Omega^2$, and $\Omega^3$ following the same definitions as those in the proof of Lemma~\ref{lem:fomkp-sufficient-aggregate}.
	
	{\bf Case I: $\cali\in\Omega^1$.} This case is the same as Case I in the proof of Lemma~\ref{lem:fomkp-sufficient-aggregate}. We can have $\opt(\cali)/\alg(\cali) = 1, \forall \cali\in\Omega^1$.
	
	{\bf Case II: $\cali\in\Omega^2$.} In this case, all knapsacks can be fully occupied in the offline solution under the worst-case instance. Thus, we can use the dual objective~\eqref{p:dual-fomkp-separable} in the \opd analysis.
	Particularly, we construct a feasible dual solution as
	$\hat\la_m = \phi_m(w_m^{(N+1)})$ and $\hat\mu_n = \mu_n^*$.
	Based on weak duality, we have
	\begin{subequations}
		\begin{align}\nonumber
			\opt(\cali) & \le\sum\nolimits_{n\in\caln}\sum\nolimits_{m\in\calm} h_{nm}(\mu_n^* + \phi_m(w_m^{(N+1)})) + \sum\nolimits_{m\in\calm} \phi_m(w_m^{(N+1)}) C_m + \sum\nolimits_{n\in\caln} \mu_{n}^* D_n\\
			\label{eq:opd-separable-eq1}
			&\le\sum\nolimits_{n\in\caln}\sum\nolimits_{m\in\calm} g_{nm}(y_{nm}^*) + \sum\nolimits_{m\in\calm} [\phi_m(w_m^{(N+1)})C_m - \sum\nolimits_{n\in\caln}\phi_m(w_m^{(n+1)})  y_{nm}^*]  \\ 
			\label{eq:opd-separable-eq2}
			&\le \sum\nolimits_{n\in\caln}\sum\nolimits_{m\in\calm} g_{nm}(y_{nm}^*) + \sum\nolimits_{m\in\calm} [\phi_m(w_m^{(N+1)})C_m - \smallint_{0}^{w_m^{(N+1)}}\phi_m(u)du]\\\nonumber
			&\le \sum\nolimits_{n\in\caln}\sum\nolimits_{m\in\calm} g_{nm}(y_{nm}^*) + (\alpha - 1)\sum\nolimits_{m\in\calm}\smallint_{0}^{w_m^{(N+1)}}\phi_m(u)du\\\nonumber
			&\le \alpha \sum\nolimits_{n\in\caln}\sum\nolimits_{m\in\calm} g_{nm}(y_{nm}^*) = \alpha \alg(\cali).
		\end{align}
	\end{subequations}
	By applying Proposition~\ref{pro:conjugate-function-separable} and KKT conditions of the problem~\eqref{p:FOMKP-utility-max-separable}, we have the inequality~\eqref{eq:opd-separable-eq1}. 
	If the threshold function $\phi$ satisfies the differential equation~in \eqref{eq:fomkp-sufficient-separable}, the inequality~\eqref{eq:opd-separable-eq2} holds.   
	The following inequalities can be easily verified based on arguments used in previous proofs. 
	Thus, if the differential equation in~\eqref{eq:opd-fomkp-sufficient} is satisfied, we have $\opt(\cali)/\alg(\cali) \le \alpha, \forall \cali\in\Omega^2$. 
	
	{\bf Case III: $\cali\in\Omega^3$.}
	The total amount of items assigned to knapsacks in $\calm^1$ is still limited in the offline solution under the worst-case instance. However, the additional items that can be reassigned from knapsacks in $\calm^2$ to those in $\calm^1$ are up to $\sum_{m\in\calm^2} w_m^{(N+1)}$. This is because in this case, each knapsack corresponds to an individual value function and thus the marginal utility (i.e., the marginal value of items minus the marginal cost of using knapsack when assigning a small bit of items) of assigning items to knapsacks in $\calm^2$ can be larger than to those in $\calm^1$ even through the marginal cost of knapsacks in $\calm^1$ is the lowest value $L$. 
	To construct a tighter upper bound of the offline optimum, we add one more constraint to the offline formulation~\eqref{p:fomkp-separable}. 
	\begin{align}\label{eq:opd-new-constraint2}
		\sum\nolimits_{n\in\caln}\sum\nolimits_{m\in\calm^1}y_{nm} \le \sum\nolimits_{m\in\calm} w_m^{(N+1)}.
	\end{align}
	The dual problem of the new offline formulation is
	\begin{align}\label{p:dual-fomkp-separable-new}
		\min_{\la_m\ge0,\mu_n\ge0,\eta \ge 0}\quad \sum\nolimits_{n\in\caln}&\left[\sum\nolimits_{m\in\calm^1} h_{nm}(\mu_n + \la_m + \eta)+ \sum\nolimits_{m\in\calm^2} h_{nm}(\mu_n + \la_m)\right] \\\nonumber
		& + \sum\nolimits_{m\in\calm} \la_m C_m + \sum\nolimits_{n\in\caln} \mu_{n} D_n + \eta \sum\nolimits_{m\in\calm}w_m^{(N+1)},
	\end{align}
	where $\eta$ is the dual variable corresponding to the new constraint~\eqref{eq:opd-new-constraint2}.
	In this case, the feasible dual variable is chosen as
	\begin{align*}
		\hat\la_m =
		\begin{cases}
			\phi_m(w_m^{(N+1)}) & m\in\calm^2\\
			0 & m\in\calm^1
		\end{cases}, \quad \hat\mu_n = \mu_n^*, \quad \hat\eta = L. 
	\end{align*}
	Then we can have
	\begin{align}\nonumber
		\opt(\cali) & 
		\le\sum\nolimits_{n\in\caln}\sum\nolimits_{m\in\calm} [g_{nm}(y_{nm}^*) - (\phi_m(w_m^{(n+1)}) + \mu_n^*) y_{nm}^*] \\
		\nonumber
		&\quad\quad\quad\quad\quad+ \sum\nolimits_{m\in\calm^2} \phi_m(w_m^{(N+1)}) C_m + \sum\nolimits_{n\in\caln} \mu_{n}^* D_n + L \sum\nolimits_{m\in\calm}w_m^{(N+1)} \\  
		\nonumber
		&\le \sum\nolimits_{n\in\caln}\sum\nolimits_{m\in\calm} g_{nm}(y_{nm}^*) + \sum\nolimits_{m\in\calm^2} [\phi_m(w_m^{(N+1)})C_m - \smallint_{0}^{w_m^{(N+1)}}\phi_m(u)du +  L w_m^{(N+1)}]\\
		\label{eq:opd-separable-eq3}
		&\le \sum\nolimits_{n\in\caln}\sum\nolimits_{m\in\calm} g_{nm}(y_{nm}^*) + (\alpha - 1)\sum\nolimits_{m\in\calm^2}\smallint_{0}^{w_m^{(N+1)}}\phi_m(u)du\\
		\nonumber
		&\le \alpha \sum\nolimits_{n\in\caln}\sum\nolimits_{m\in\calm} g_{nm}(y_{nm}^*) = \alpha \alg(\cali).
	\end{align}
	The key step in the \opd analysis above is to ensure the inequality~\eqref{eq:opd-separable-eq3}. 
	If the threshold function $\phi$ satisfies the differential equation~\eqref{eq:fomkp-sufficient-separable}, the inequality~\eqref{eq:opd-separable-eq3} holds and we have $\opt(\cali)/\alg(\cali) \le \alpha, \forall \cali\in\Omega^3$.
	
	In summary, the competitive ratio of $\ota_\phi$ for the \fomkp with separable value functions is $\alpha$ if the threshold function $\phi$ satisfies the sufficient condition in Lemma~\ref{lem:fomkp-sufficient-separable}.
\end{proof}

Using Lemma~\ref{lem:fomkp-sufficient-separable} to complete the proof, is similar to the case of aggregate functions.  we determine the threshold function that satisfies the differential equation~\eqref{eq:fomkp-sufficient-separable} and can minimize $\alpha$. Applying Gronwall's Inequality to \eqref{eq:fomkp-sufficient-separable} gives
\begin{align*}
	\varphi_m(w) \le \frac{\alpha L\beta_m}{C_m} - \frac{L}{C_m}w + \frac{\alpha}{C_m} {\smallint_{\beta_m}^{w}}\varphi_m(u) du = \frac{L}{\alpha} + [\frac{(\alpha - 1)L\beta_m}{C_m} - \frac{L}{\alpha}] e^{\alpha(w - \beta_m)/C_m}, w\in[\beta_m, C_m]
\end{align*}
Since $\varphi_m(C_m) \ge U$, we have $U\le \varphi_m(C_m) \le \frac{L}{\alpha} + [\frac{(\alpha - 1)L\beta_m}{C_m} - \frac{L}{\alpha}] e^{\alpha(C_m - \beta_m)/C_m}$. The minimal $\alpha$ is achieved when all inequalities in \eqref{eq:fomkp-sufficient-separable} are binding. We then have $U = \frac{L}{\alpha} + [\frac{(\alpha - 1)L\beta_m}{C_m} - \frac{L}{\alpha}] e^{\alpha(C_m - \beta_m)/C_m}$ and $\beta_m = \frac{C_m}{\alpha - 1}$. Thus, the minimal competitive ratio is the solution of the equation $\alpha_{\phi^*} - 1 - \frac{1}{\alpha_{\phi^*} - 1} = \ln\frac{\alpha_{\phi^*} \theta - 1}{\alpha_{\phi^*} - 1}$ and the threshold function is given by \eqref{eq:threshold-fomkp-separable}.


\end{document}